%% file: main.tex
\documentclass[10pt, compsoc]{IEEEtran}
\usepackage[ruled,linesnumbered]{algorithm2e}
\usepackage{multirow}
\usepackage{balance}
\usepackage{booktabs}
\usepackage{ragged2e}
\usepackage{tabularx}
\usepackage{graphicx} 
\usepackage{subcaption}
\usepackage[numbers, sort&compress]{natbib}
\usepackage{svg}
\usepackage{url} 
\usepackage{amsmath}
\usepackage{amsfonts, amssymb}
\usepackage{expl3}
\usepackage{soul}  
\usepackage{amsthm}
\usepackage{mathtools}
\usepackage{etoolbox}
\usepackage{colortbl}
\usepackage[T1]{fontenc}
\usepackage{dblfloatfix} 
\usepackage{xparse}
\usepackage{xstring}
\usepackage{bbm}
\usepackage{changepage}
\usepackage{framed}

\newcolumntype{L}{>{\RaggedRight\hangafter=1\hangindent=1em}X}

\makeatletter
 {\par\unskip\endMakeFramed}
\makeatother

\definecolor{formalshade}{rgb}{0.93,0.93,0.93}
\definecolor{darkblue}{rgb}{0.2, 0.2, 0.2}

\newenvironment{formal}{%
  \def\FrameCommand{%
    \hspace{1pt}%
    {\color{darkblue}\vrule width 2pt}%
    {\color{formalshade}\vrule width 4pt}%
    \colorbox{formalshade}%
  }%
  \MakeFramed{\advance\hsize-\width\FrameRestore}%
  \noindent\hspace{-1pt}
  \begin{adjustwidth}{}{7pt}%
  \vspace{2pt}\vspace{2pt}%
}
{%
  \vspace{3pt}\end{adjustwidth}\endMakeFramed%
  }

\newtheorem{theorem}{Theorem}
\newtheorem{corollary}{Corollary}[theorem]
\newtheorem{lemma}{Lemma}

\newcommand{\bi}{\begin{itemize}}
	\newcommand{\ei}{\end{itemize}}
\newcommand{\be}{\noindent\begin{enumerate}}
	\newcommand{\ee}{\noindent\end{enumerate}}

\DeclarePairedDelimiter\abs{\lvert}{\rvert}%
\DeclarePairedDelimiter\norm{\lVert}{\rVert}%
\makeatletter
\let\oldabs\abs
\def\abs{\@ifstar{\oldabs}{\oldabs*}}
\let\oldnorm\norm
\def\norm{\@ifstar{\oldnorm}{\oldnorm*}}
\makeatother

\newcommand{\graycolor}{gray!20}
\newcommand{\gray}[1]{\cellcolor{\graycolor}\textbf{#1}}

\usepackage{enumitem}
\setitemize{noitemsep,topsep=0pt,parsep=0pt,partopsep=0pt,leftmargin=*}

\DeclareMathOperator*{\argmax}{arg\,max}
\DeclareMathOperator*{\argmin}{arg\,min}

\newcommand{\IT}{\texttt{SMOOTHIE}}
\newcommand{\smoothness}{smoothness}

\newcommand{\respto}[1]{
\fcolorbox{black}{black!15}{%
\label{resp:#1}%
\bf\scriptsize R{#1}}}

\newcommand{\BLUE}{\ifshowchanges \color{blue} \fi}
\newcommand{\BLACK}{\color{black}}

\ExplSyntaxOn

\NewDocumentCommand{\citeresp}{m}
{
  (see~
  \seq_set_split:Nnn \l_citeresp_items_seq {, } { #1 } 
  \seq_pop_right:NN \l_citeresp_items_seq \l_lastitem_tl 
  \seq_map_function:NN \l_citeresp_items_seq \__citeresp_format:n 
  \__citeresp_formatnocomma:n \l_lastitem_tl
  )
}

\cs_new_protected:Nn \__citeresp_format:n
{
  \fcolorbox{black}{black!15}{\bfseries\scriptsize R#1}~on~page~\pageref{resp:#1},~
}

\cs_new_protected:Nn \__citeresp_formatnocomma:n
{
  \fcolorbox{black}{black!15}{\bfseries\scriptsize R#1}~on~page~\pageref{resp:#1}
}

\seq_new:N \l_citeresp_items_seq 

\ExplSyntaxOff

\newenvironment{response}[2]{
    \par\noindent
    \BLUE 
    \respto{#1} {#2}%
}{
    \par\noindent
    \BLACK
}

\newif\ifshowchanges
\showchangesfalse
\ifshowchanges
    \newcommand{\changed}[2]{\BLUE \respto{#1} {#2} \BLACK}
\else
    \newcommand{\changed}[2]{#2}
\fi

\begin{document}

\title{
     Is Hyper-Parameter Optimization \\ Different for Software Analytics?
}

\author{Rahul Yedida, Tim Menzies,~\IEEEmembership{Fellow, IEEE}
\IEEEcompsocitemizethanks{\IEEEcompsocthanksitem R. Yedida and  T. Menzies are with the Department
of Computer Science, North Carolina State University, Raleigh, USA.
 \protect\\
E-mail: ryedida@ncsu.edu, timm@ieee.org
}
}

\markboth{IEEE Transactions on Software Engineering}%
{Yedida \MakeLowercase{\textit{et al.}}: {\smoothness} finds better hyper-parameter options}

\IEEEtitleabstractindextext{
\begin{abstract}



Yes. SE data can have ``smoother'' boundaries between classes (compared to traditional AI data sets). To be more precise, the magnitude of the second derivative of the loss function found in SE data is typically much smaller. A new hyper-parameter optimizer, called {\IT}, can exploit this idiosyncrasy of SE data. We compare {\IT} and a state-of-the-art AI hyper-parameter optimizer on three tasks: (a) GitHub issue lifetime prediction (b) detecting static code warnings false alarm; (c) defect prediction. For completeness, we also show experiments on some standard AI datasets. {\IT} runs faster and predicts better on the SE data--but ties on non-SE data with the AI tool. Hence we conclude that SE data can be different to other kinds of data; and those differences mean that we should use different kinds of algorithms for our data. To support open science and other researchers working in this area, all our scripts and datasets are available on-line at \url{https://github.com/yrahul3910/smoothness-hpo/}.
\end{abstract}
}

\maketitle

\IEEEpeerreviewmaketitle

\section{Introduction}
To say the least, much current work in SE uses
algorithms from   the AI community. Usually, these algorithms are used in their default ``off-the-shelf'' configuration~\cite{Fu16ist,xu15fse}.  

Is this unwise? Is software engineering
data  fundamentally the same as those used to develop those AI tools? Or does software engineering datasets have unique characteristics   that makes it different from others, and therefore would benefit from other kinds of AI tools?

This article argues that it is \underline{\bf not} always
advisable to apply standard AI tools to SE data.
Specifically, in the case of hyper-parameter optimization (HPO), we show that:
\bi 
\item  {\em AI for SE is different to standard AI}. SE data can have much  ``smoother'' boundaries between classes. To be  
precise, compared to AI data,  the magnitude of the second derivative of the loss function found in SE data can be  much smaller.
\item  For SE data, {\em better HPO result can be achieved using tools designed for SE data}. 
\ei
This is an important result since HPO is an important problem.
Hyper-parameter optimizers    explore the myriad of configuration choices which control either the  learning  or the data pre-processing.
As shown by
 Table~\ref{tab:hyp:eg} and Table~\ref{what}, there are many such hyper-parameters.
 Much SE research has shown that HPO  is very   effective at improving the performance of models built from SE data~\cite{fu2016tuning, agrawal2019dodge, yedida2021value, yedida2023find}; for example, see   Figure~\ref{tant}.
Nevertheless, while  studying industrial text-mining tools, we  have found  many examples of poorly-chosen hyper-parameters (which
 meant our industrial partners were selling less-than-optimal  products)~\cite{rahul16fse}.

\begin{figure*}[!t]
\begin{center} \includegraphics[width=.6\linewidth]{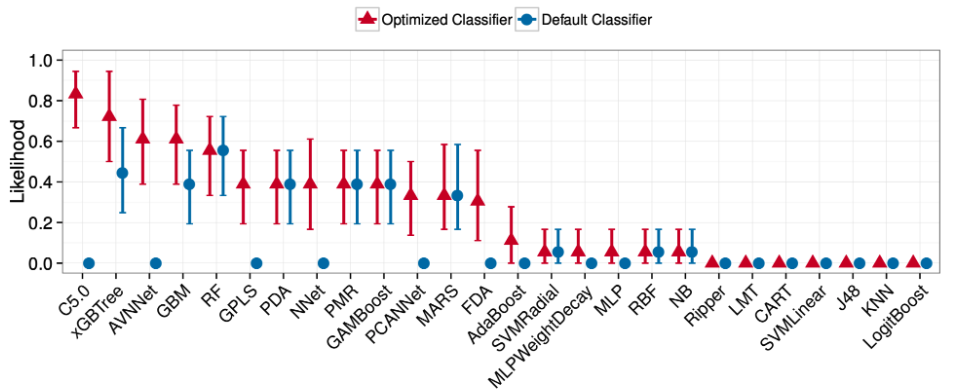}\end{center}
\caption{Defect prediction: likelihood of a learner performing best. BLUE shows   performance pre-optimization and RED shows very large  performance
 improvements (post-optimization). 
E.g. on the left-hand-side,
 the C5.0 decision tree learner had {\em worst} performance (pre-optimization) and {\em best} performance (post-optimization). From~\cite{Tantithamthavorn16}.}\label{tant}
\end{figure*}

This paper offers a new hyper-parameter optimizer called {\IT} learned from SE data that selects hyper-parameter options according to how well those options
simplify the boundary between (say) defective and non-defective modules. For this paper, the important features of {\IT} are that:
\bi 
\item For SE data, {\IT}  makes better predictions than   state-of-the-art-methods taken from AI; 
\item {\IT}  runs much faster than that  state-of-the-art;
\item But {\IT} performs poorly   for non-SE data.
\ei
To say all that another way: 
\begin{quote}
    {\em SE data can be different from other kinds of data; and those differences mean that we should use different kinds of algorithms for our data.} 
\end{quote}

\changed{Ea1.1}{
    To summarize, the core contributions of this paper are:

    \begin{itemize}
        \item We extend empirical observations from prior work on static code warnings, which showed that SE loss landscapes can be ``smoother'' than non-SE landscapes, to show that this is a domain-specific idiosyncrasy of SE data. We connect this observation with intuitive explanations for why SE data exhibits this property, while standard AI data \textit{cannot}. The crux of this paper is that this property is specific to SE data, and that SE research should devise algorithms that leverage such unique properties.
        \item We exploit this property to devise a hyper-parameter optimization method tailored to SE data. Our method uses a very cheap probing computation to test the smoothness of loss landscapes, discarding configurations that are unpromising. We motivate this method with recent work from the AI literature showing that such ``smooth'' landscapes are desirable for generalization to unseen data.
        \item We derive semi-empirical equations to compute the smoothness in a computationally cheap manner for different learning algorithms. We then experiment on three SE tasks (for a total of 19 datasets), showing that our method outperforms standard AI methods, while also being faster. Further, we show that the assumption underpinning this method (that of ``smooth'' loss landscapes) does not hold in general for AI datasets--which means standard AI methods may be preferable in such scenarios.
    \end{itemize}

    However, it is important that the reader be aware of the limitations of our method (these are expanded upon in the Limitations section at the end of this paper):
    \begin{itemize}
        \item As presented here, our method relies on deriving \textit{learner-specific equations} for smoothness. In future work, we will extend this to more learning algorithms and also devise a learner-agnostic method; but that is out of scope of this paper.
        \item The above limitation leads to some complications in implementations. We provide a reference implementation for our method, which uses \textit{function overloading} to use the type of the classifier passed to select the corresponding equation.
    \end{itemize}
}

The rest of this paper is structured as follows. The next section 
discusses how SE data is different to AI data.
This is followed by 
notes on related work, specifically, hyperparameter optimization and its role in software analytics. 
 Section \ref{sec:experiments} discusses our experimental setup, and the results of those experiments are presented in Section \ref{sec:results}. We discuss why {\IT} works so well for SE data along with threats to validity in Section \ref{sec:discussion}, before discussing limitations in Section \ref{sec:limitations} and future work in Section \ref{sec:futurework}.




\section{Unique Properties of SE Data}
Our premise is that 
SE needs its own kind of AI since data from SE projects is somehow fundamentally different (than the data used to develop standard AI algorithms).
This section offers evidence for that premise, in two parts.
Firstly we will argue that we {\em expect} SE data to be smooth. Secondly, we will also say that other
kinds of data may not be as smooth.

\subsection{Expect Smoothness in SE Data}\label{sec:whydiff}

A repeated result in SE is that SE artifacts are surprisingly regular.
\citet{Hindle12} writes that ``programming languages,
in theory, are complex, flexible and powerful, but the programs that real
people write are ... rather repetitive, and thus they have
statistical properties that can be captured''.

There are many   possible reasons for this repetition.
For example, 
successful software is {\em  maintainable}; i.e. 
  it   contains
 structures {\em recognizable} to those
 doing the maintenance. 
  Hence there are selection pressures forcing those 
 who write software to     make repeated   use of:
 \begin{itemize}
 \item Readily identified   idioms such as   for-loops;
 \item Or higher level patterns such as subject-observer~\citep{gamma1993design};
 \item
 Or even higher-level architectural patterns such as the microservices architecture discussed by \cite{Chen18}.  
 \end{itemize}
Such repetition means that the state space of software
  can contain   many repeated structures-- which would mean that  different parts of a code base may be surprisingly similar. Data generated from such  repetitive structures might itself be repetitive (with minimal differences between  adjacent clusters in the data).

\begin{figure}[!t]
\begin{center}
  \includegraphics[width=1in]{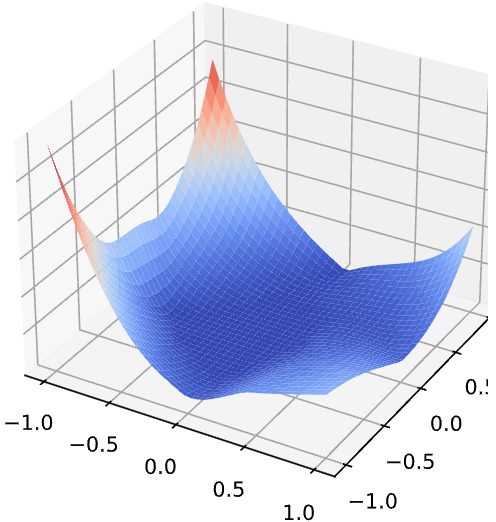} ~~~~~~~~~~\includegraphics[width=1.4in]{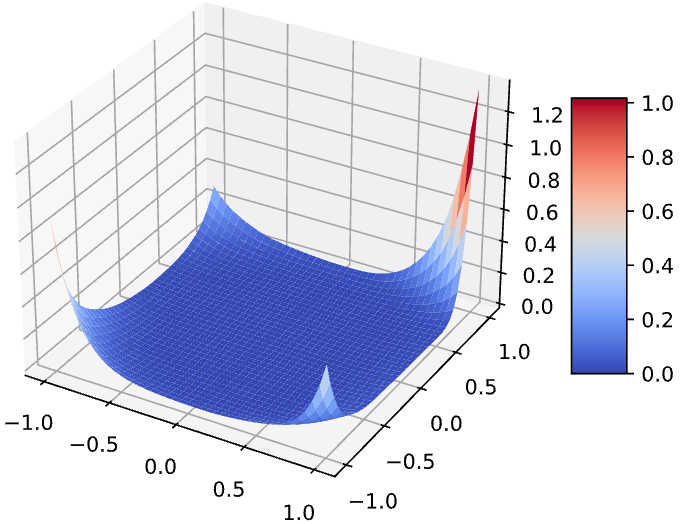}
  \end{center}
        \caption{Left \& right shows loss landscape before \& after HPO.
The loss function  be visualized as the $z$
dimension of  domain data   mapped into two PCA components $x,y$.   In this figure, red indicates  regions with highest cross-entropy error.} \label{fig:loss}
    \end{figure}
    
  Whatever the reason for this repetition, that regularity would mean there are few differences between adjacent parts of the code. 
  Such similarities would manifest
  as smooth surfaces within the data. Hence, for some time now,
  the authors have been asking if
we   can   {\em see} such smoothness and can we 
  {\em use} that smoothness to achieve some useful results. 

It was therefore   gratifying, in Fall 2023, to actually observe such smoothness. While working on static code analysis false alarm identification  and hyper-parameter optimization~\cite{yedida2023find},
 one of us
 (Yedida)  noticed a pattern in changes to the 
 {\em loss functions} before and after hyper-parameter optimization.
Loss functions are surrogates for metrics that have gradients (and the better the learner, the lower the loss, generally speaking). 
Figure~\ref{fig:loss} left and right
shows that:
\begin{quote}
    {\em After optimization, the  loss function gets ``smoother''  (i.e. on average, fewer changes in any local region).} \end{quote}

 \noindent To say that another way, a learner that respects the smoothness landscape of SE data does better than otherwise.

\begin{table}[!b]
    \caption{Different kinds of landscapes}
    \label{tab:landscapes}

    \begin{tabularx}{\linewidth}{|L|}
        \hline   
        \textbf{Loss landscape:} Loss landscapes are the concept this paper predominantly talks about, and refers to the topography of the loss function of a deep learner. This is a function of the model's parameters, and therefore exists in a high-dimensional space. \\
        \hline
       \textbf{Hyper-parameter landscape:} The hyper-parameter landscape \cite{pushak2022automl, huanghyperparameter} is the topography of a loss function whose inputs are the hyper-parameters chosen. This function typically involves the model's loss function, and in an empirical risk minimization formulation can be written as a bilevel optimization problem \\
      
        \begin{center}
           $
            \begin{aligned}
            h^* &= \argmax\limits_{h \sim \mathcal{H}} \mathbb{E}_\mathcal{X \sim D}\left[ f(\theta^*(h))  \right] \\
            \theta^*(h) &= \argmin\limits_{\theta \in \mathbb{R}^d} \mathcal{L}(X; \theta, h)
            \end{aligned}
        $
        \end{center}
        
        where $f$ is a performance measure (such as accuracy), and for brevity, we have omitted that the model is typically evaluated on a different dataset held out from the training set.

        While the hyper-parameter optimization problem itself is over this landscape, we use a random search to explore it in this paper, and for the chosen hyper-parameters at each step, we use the smoothness to decide if it is worth exploring further.
        \\ \hline
    \end{tabularx}
\end{table}

\changed{3a3.1}{
 (Aside: it should be noted that the machine learning literature distinguishes ``loss landscapes'' from the ``hyperparameter landscape'', which is the loss value distribution over the space of hyper-parameters). So, to be clear, when we say ``landscape'' we are referring to the loss landscape (for more on this, please  see Table~\ref{tab:landscapes}).
 }

 Based on that observation, we theorized that:
\begin{quote}
    {\em Hyper-parameter optimization is best viewed as a ``smoothing'' function for the
decision landscape.}
\end{quote}

This theory predicts that the greater the smoothing, the  better 
the hyper-parameter optimization. This paper tests that prediction.
The algorithm we call {\IT} is an experiment on whether or not increasing smoothness could be used to guide hyper-parameter optimization (HPO).
As shown in this paper,  
 {\IT}  can assess hyper-parameter options {\em without needing
to run a learner to completion}. For example, for
neural networks that learn over 100 epochs, {\IT} can assess hyper-parameter options within the first epoch. That is, {\IT} can focus on good configurations by very quickly rejecting bad configurations using a fraction of the CPU cost of conventional hyper-parameter optimizers.

\subsection{Other Kinds of Data May be Less Smooth}

We argued above that SE  practices are designed to generate maintainable software, and that such practices 
may also generate SE data that exhibits smoothness.
Data from some non-SE  sources  is not constrained in the same way. Hence,
we would expect non-SE data to be less smooth than SE data. 

This is actually the case.
Figure~\ref{fig:beta-dist} shows the smoothness measured in the data 
used in this study. This figure was calculated using some equations 
described later in this paper.  In summary, smoothness is a measure of the magnitude of the second
derivative in the loss function. {\em Lower}   smoothness means 
{\em more  regularity} and {\em repeated properties} between those adjacent parts\footnote{Although unintuitive, the mathematical term ``smoothness'' is a bound on the \textit{sharpness} as opposed to the intuitive smoothness. We will make this precise later.}.

\begin{figure}[!t]
        \includegraphics[width=\linewidth]{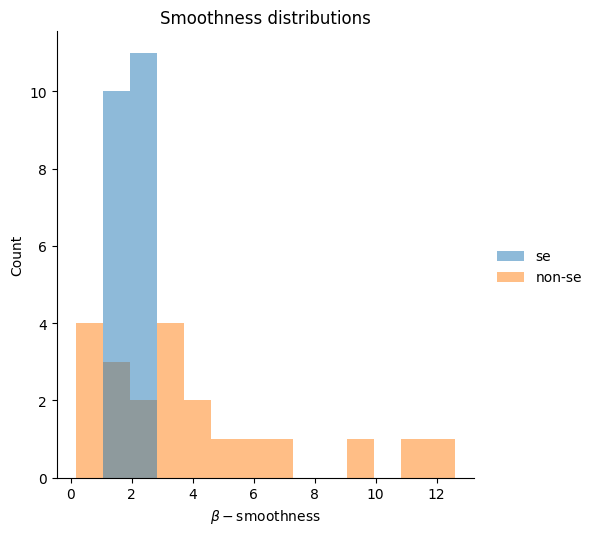}
        \caption{Distribution of smoothness for SE and non-SE datasets. 
        Calculated using the maths
        of  \S\ref{fig:beta-dist}.
        See \S\ref{analytics} for the data used in this figure.
For details on how the x-axis values are calculated, see \S\ref{about}.
Note the mean values for the two populations are different:
mean smoothness = 2, 5 for SE, AI respectively. }
        \label{fig:beta-dist}
    \end{figure}

In Figure~\ref{fig:beta-dist},
blue shows the smoothness of the loss function for the SE data mentioned in the last section. 
Also in that figure, orange 
shows the smoothness seen in various AI data sets from UCI
(the data sets here were selected from prior work that also explored the value of different algorithms
for SE and non-SE dats~\cite{agrawal2021simpler}).

The stand-out feature of this graph is that the SE data has dramatically different smoothness to the non-SE data. 
This paper will   show that the x-axis of this figure can predict  what algorithm will be most successful:
\bi
\item At high smoothness values, {\IT} loses to standard AI algorithms.
\item  But, as we shall show,
{\IT} defeats them
at low smoothness values. 
  \ei
Clearly, this result supports the claim in the introduction that (a)~SE data can be different to standard AI data so (b)~it is unwise to always assume that  AI algorithms should be applied to SE data.

\section{  Hyper-parameter Optimization (HPO)}\label{aboutHPO}

We said above that low/high smoothness can predict for our new
algorithm winning/losing against standard AI algorithms.
The task in that experimental comparison was {\em hyperparameter optimization} (HPO). This section
defines HPO, and offers examples of how it is used in SE. Afterwards, the rest of this paper uses HPO to show the value of smoothness.

\begin{table} 
{\small
\caption{A sample hyper-parameter choices  for 
\underline{\bf learners} (note, this is not a complete set).}
\label{tab:hyp:eg}
\begin{tabular}{|p{3.3in}|}\hline
\rowcolor{blue!5}
{\em  Feedforward networks} might have (say) 3 to 20 layers and 2 to 6 units per layer
\\\hline
 {\em Deep learners} have   many more   control parameters 
including how to set their topology and activation functions. 
\\\hline \rowcolor{blue!5}
{\em Neighbor classifiers} need   to know
 (a)~how   to measure
distance; (b)~how many neighbors to use; (c)~how to combine   those neighbors;
(d)~how (or if) to pre-cluster
the data (to optimizer th neighborhood lookup); etc. 
\\\hline
{\em Random forests} need to know how many trees to build, how many features to use in each tree, and how to  poll the whole forest (e.g., {\em majority} or {\em weighted majority}).\\\hline \rowcolor{blue!5}
Decisions for {\em logistic regression} include which
regularization terms ($L_1$, $L_2$) to use, (or, indeed, to  use both) and (b)~the strength of the regularization; i.e.  $C \in \{0.1, 1, 10\}$. \\\hline
\end{tabular}}
\end{table}
\begin{table} 
\caption{A sample of    hyper-parameter choices for 
\underline{\bf data pre-processing} (note, this is not a complete set).}\label{what}
\small\begin{tabular}{|p{.95\linewidth}|}\hline
\rowcolor{blue!5}
Some of   pre-processors are simple:
\bi 
\item normalization: $x \to \frac{x}{\lVert x \rVert}$; 
\item standardization: $x \to \frac{x - \mu_x}{\sigma_x}$;
\item min-max scaling: $x \to \frac{x-\min x}{\max x - \min x}$; 
\item max-absolute value scaling: $x \to \frac{x}{\max \lvert x \rvert}$); 
\item robust scaling: $x \to \frac{x - P_{50}(x)}{P_{75}(x) - P_{25}(x)}$ where $P_i(x)$ represents the $i$th percentile of x.
\ei
Various authors recommend one or more of these simpler pre-processors. For example, for instance-based methods \citet{aha91} recommend min-max scaling since it means numbers on different scales can be compared.
\\\hline
Another widely used pre-processor is SMOTE~\cite{chawla2002smote}. In order to handle unbalanced training data, SMOTE generates synthetic minority samples by extrapolating along line segments joining the nearest neighbors of the minority class.
\\\hline
 \rowcolor{blue!5}
``Fuzzy sampling'' \cite{yedida2021value},  concentrically adds $\frac{1/n}{2^i}$ samples at concentric layer $i$, for $\left\lfloor{\log_2{1/n}}\right\rfloor$  layers, where $n$ is the class imbalance ratio. Unlike other oversampling techniques, fuzzy sampling reverses the class imbalance. The authors show that applying it twice can also be beneficial. 
 \\\hline
Another   pre-processor from semi-supervised learning~\cite{yedida2023find} (which we call ``label engineering'') is to delete, at random $m - \sqrt{m}$ labels. These are then recovered by building a kd-tree \cite{bentley75} on the remaining $\sqrt{m}$ samples. For each lost label, the kd-tree is queried for the nearest $\sqrt[4]{m}$ samples, and the mode of their labels is assigned.
 \\\hline
\end{tabular}

\end{table}

\subsection{About HPO}

A learning system combines some data pre-processing and machine learning.
Every such system has magic control parameters that adjust the data pre-processing and/or the machine learning
(e.g. see Table~\ref{tab:hyp:eg} and Table~\ref{what}). Even something as simple
as the $k-$nearest neighbors classifier needs to decide how many neighbors
it should study in order to classify a new example (and, as seen in Table~\ref{tab:hyp:eg}, it must settle on several other configuration options). 

In essence, HPO is the process of fiddling with the control
parameters of a learner and the data pre-processing. The trivial way to do this is {\em grid search}, which is a set of nested for-loops across all possible values of all the control parameters. In practice, grid search can be impractically slow and ineffective~\cite{bergstra2012random}.

\begin{figure*}[!t]
    \includegraphics[width=\textwidth]{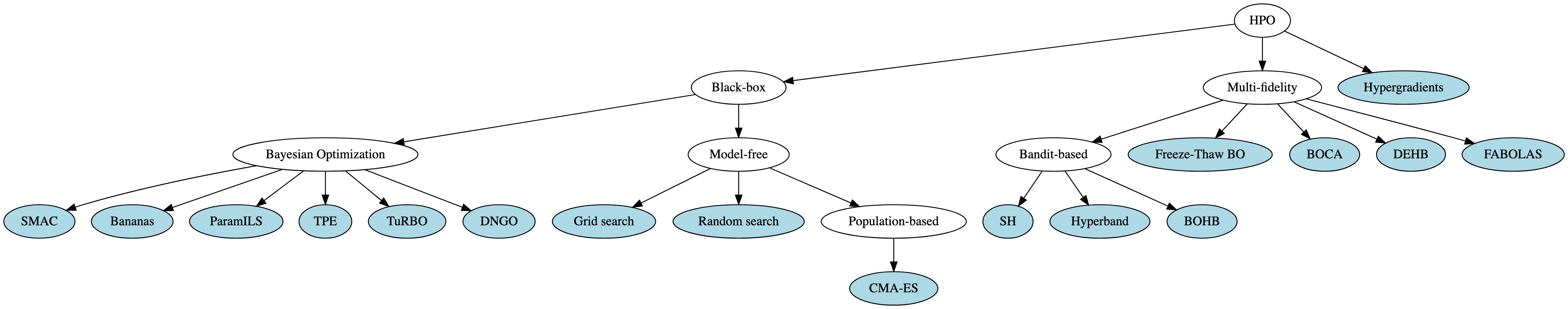}
    \caption{A brief overview of hyper-parameter optimization methods. For a more in-depth review, see \cite{shahriari2015taking, feurer2019hyperparameter, bischl2023hyperparameter}.}
    \label{fig:hpo:graph}
\end{figure*}

Smarter HPO algorithms take more care with how they explore the hyper-parameter space.  For example,
\cite{bergstra2011algorithms,bergstra2012random} discuss methods to reduce the cost of hyper-parameter optimization via some clever random sampling. They demonstrate the effectiveness of random search under a limited computational budget, additionally showing that for most problems, only a few hyper-parameters matter, although that subset of important hyper-parameters is different across datasets. To introduce the prior state-of-the-art, we now discuss two more strategies to explore the hyper-parameter space more efficiently.

Bayesian optimization \cite{movckus1975bayesian} is a sequential model-based optimization algorithm that improves as it gathers more data. It involves a {\em surrogate modeling} for the objective function and an acquisition function that samples configurations.  Surrogates are a model learnt from just a few examples  which are used to assess all the examples.  
Surrogates   serve as a cheap   ``acquisition functions'' to find interesting examples
of what to evaluate next.

A widely-used surrogate method is Gaussian Process (GP) models \cite{williams2006gaussian}. Various acquisition functions have been proposed, such as Expected Improvement \cite{jones1998efficient}, Gaussian Process Upper Confidence Bound (GP-UCB) \cite{lai1985asymptotically}, and predictive entropy search \cite{hernandez2014predictive}. Among these, Expected Improvement is the most popular choice. Of the various BO-based methods, BOHB \cite{falkner2018bohb} combines Bayesian optimization and Hyperband\footnote{Briefly, Hyperband is based on ``successive halving'', a method that eagerly discards unpromising candidates at different computational budgets (so-called ``fidelities'')--hence the term ``multi-fidelity'' is used to describe it.} \cite{li2017hyperband}, and has emerged as the most popular surrogate-based algorithm combining these two due to it being open-source and state-of-the-art. The BO closely resembles TPE \cite{bergstra2011algorithms} with a single multi-dimensional kernel density estimate \cite{yu2020hyper}. In a large-scale study, \citet{eggensperger2021hpobench} endorsed multi-fidelity methods (such as BOHB) over their black-box counterparts in constrained budget settings. This result was corroborated by a more recent work by \citet{pfisterer2022yahpo}, which showed that in constrained budget settings, multi-fidelity methods outperformed full-fidelity methods such as SMAC \cite{lindauer2022smac3}. 


    Another strategy for searching in a more informed manner is \textit{differential evolution} (DE). This is an evolutionary method that maintains a population of candidate solutions at each generation, and uses three randomly chosen individuals to build a new candidate for the next generation. In fact, DE has already seen applications in the SE literature: \citet{majumder2018500+} showed that using DE to tune fast classical learners proved to be over 500 times faster than deep learning, \textit{while outperforming it}; and \citet{agrawal2018better} showed that using DE to tune SMOTE \cite{chawla2002smote} (an oversampling algorithm) can significantly improve the performance of many learning algorithms. DEHB \cite{awad2021dehb} uses DE on top of Hyperband. On top of improving its performance, it allows for a distributed tuning of hyper-parameters.

HPO is a very large area of research and this section has only sampled some of that work. Figure \ref{fig:hpo:graph} shows a  broad classification of HPO algorithms, along with a few notable examples.
For further details on HPO, see  \cite{shahriari2015taking,feurer2019hyperparameter,bischl2023hyperparameter}.  

Many studies from the AI literature endorse BOHB and DEHB as the current state-of-the-art~\cite{yu2020hyper, yang2020hyperparameter, feurer2019hyperparameter}. 
Hence, when we need to  compare {\IT} against some state-of-the-art AI HPO algorithm, we will use these two, along with random search as a simpler baseline.
 
\subsection{HPO in software analytics}\label{analytics}

  Here,
we offer details of some case studies with SE and HPO. The following list is hardly exhaustive, but it does detail some of
the more recent work in this area.  Also, this section serves
to introduce the examples used later in this paper.

\subsubsection{Defect prediction.}  Software quality assurance budgets are finite while assessment effectiveness increases exponentially with assessment effort \cite{fu2016tuning}. Therefore, standard practice is to apply slower methods on code sections that seem most  {\em interesting} (most critical or bug-prone).  Defect prediction is a fast way to learn these {\em interesting regions} using static code attributes.  A 2018 study \cite{wan2018perceptions} of 395 practitioners across 33 countries found that over 90\% were willing to adopt defect prediction models in their workflow.

As shown in Figure~\ref{tant}, many of the algorithms used in defect prediction can benefit from
hyper-parameter optimization.
Recently, \citet{yedida2021value} proposed ``fuzzy sampling'' (see Table \ref{what}) and combined it with stacked autoencoders, SMOTE, hyper-parameter optimization (using DODGE \cite{agrawal2019dodge}), and a weighted loss function to achieve state-of-the-art results on defect prediction using the PROMISE repository \cite{boetticher2007promise}, which contains 20 static code attributes of the software such as the depth of inheritance tree \cite{chidamber1994metrics}, and mean McCabe cyclomatic complexity \cite{mccabe1976complexity}. In this study, we use the same 6 datasets as that study for a fair comparison.

\subsubsection{Issue lifetime prediction.} Issue lifetime predictors learn predictors for how long it takes to address issues raised by other developers (or end-users). Predicting issue close time has multiple benefits for developers, managers, and stakeholders, including task prioritization and resource allocation. The task is interesting since in open source software, developers triage their work according to the difficulty of their work; i.e. they will explore many small tasks before tackling harder tasks that are slower to accomplish.

Hyper-parameter optimization has been successfully applied  to issue lifetime prediction.
For example, recently we studied the  DeepTriage \cite{mani2019deeptriage} system. This is a state-of-the-art-tool
that uses   a bidirectional LSTM with attention after removing special characters, tokenization, and pruning. That code base is complex and took eight hours to train a predictor. Recently we found~\cite{yedida2021old} that 
  DeepTriage can be defeated using    HPO and a forty-year old feed-forward neural net.
The reason for this superior performance was simple: DeepTriage 
and that older neural net need  eight hours and three seconds (respectively) to build a predictor. The faster algorithm is hence
(much) more suitable for multiple   hyper-parameter optimization studies.  Hence the conclusion of that study was that "old, fast, and simpler algorithms (plus HPO) can defeat new, slow, and supposedly smarter algorithms''. 

This paper uses the same data sets as that DeepTriage study; i.e.  data mined from Bugzilla for three popular software projects: Firefox, Chromium, and Eclipse. For simplicity, although the data contains user comments and system records (both of which are text data sources), we do not use those features in this study for our method. \citet{lee2020continual} proposed binning the data into $k \in \{2, 3, 5, 7, 9\}$ classes and test the performance of classifiers under the binary and multi-class settings. For brevity, we show results for 2 and 3 classes.

\subsubsection{Static code analysis false alarms}
Static analysis (SA) tools report errors in source code,
without needing to execute that code
(this makes them
very popular in industry).  SA tools    generate many false alarms that are ignored by practitioners. 
 Previous research work shows that   35\% to
91\% of SA warnings reported as bugs by SA tools 
are routinely
ignored by developers~\cite{heckman2009model,heckman2008establishing}. Hence many researchers add a machine learning post-processor that reports the ``alarms you should not ignore''
(learned from the historical log of  what errors
 do/do not result in developers changing the code)~\cite{wang2018there,yang2021learning,yedida2023find}.

Initial results from  \citet{yang2021learning} suggested that it was possible  to find the actionable static code warnings generated by a static analysis tool like FindBugs, using a very simple learner.
But later work by \citet{kang2022detecting} showed that the datasets used by the former study had errors, and once manually repaired, those simple methods  did not produce good results. This was later fixed by \citet{yedida2023find}, using 
hyper-parameter optimization that explored
combinations of   pre-processing methods.  While successful,
the methods of Yedida
et al. were somewhat of a hack.
The maths and methods of this paper
replaces that hack with an approach that is much more  well-founded approach (and which is simpler to apply).

\section{About {\IT}}\label{about}

The last section described some of the case studies we will explore (defect prediction, issue lifetime prediction, and static code analysis false alarms).
The  experiments of this paper check if there is any value in turning ``bumpy'' landscapes (like Figure~\ref{fig:loss}, left-hand-side) into ``smoother'' landscapes
(like
Figure~\ref{fig:loss}, right-hand-side).
To conduct such experiments, we used the  ``smoothing'' tool described in this section.

{\IT} begins by sampling a few    configurations of a learner, selected
at random from, say, Table~\ref{what}. Then, we compute the {\smoothness} for each configuration using the methods described in this section. 
\begin{algorithm}[!t]
    \SetAlgoLined
    \SetKwInOut{KwInput}{Input}
    \SetKwInOut{KwOutput}{Output}
    \KwInput{Number of configs to sample $N_1$, number of configs to run $N_2$. Defaults: $N_1 = 30, N_2 = 5$.}
    \KwOutput{A near-optimal configuration}
    $\mathcal{H}_0 \gets$ \textsc{Random}($\mathcal{H}, N_1$)\;
    $S \gets \phi$\;
    $P \gets \phi$\;
    \For{each config $h$ in $\mathcal{H}_0$}{
        S[config] $\gets$ \textsc{Get-Smoothness}(h)
    }
    \For{each config in \textsc{Top}(S, $N_2$)}{
        P[config] $\gets$ \textsc{Run}(config)\;
    }

    \KwRet{$\argmax P$}\;
    \caption{{\IT}}
    \label{alg:smoothie}
\end{algorithm}


Algorithm \ref{alg:smoothie} shows the pseudo-code for {\IT}. We first pick $N_1$ ($=30$) configurations at random. For each of those configurations, we compute the {\smoothness} and store it. Next,we pick the top $N_2$ ($=5$) of those configurations by {\smoothness} values, and run only those, getting the performance for each. The algorithm then return the best-performing configuration.

 The rest of this section shows the \texttt{Get-Smoothness} 
   calculation used in Algorithm~\ref{alg:smoothie}.
But before that, we make the following important points.
Our definition
of smoothness is general to many learners (See Equation~\ref{brief}). However,  that definition
must be realized in different ways for different learners.
\bi 
\item 
In some cases smoothness can be computed just by looking at certain properties of the data (in Gaussian Naive Bayes, that is the covariance matrix of the data, see Theorem~\ref{th:nb}).
In this  first case, HPO is very fast--just peek at the data to rank and prune the operators
of (say) Table~\ref{what}.
\item
In other cases, as in neural networks (see Theorem~\ref{th:ff}), smoothness changes by the inference process. In this latter case, we have to watch the learner running (at least for a limited time) before smoothness can be computed\footnote{This brings the weights closer to the final weights at an exponential rate, assuming the loss is also strongly convex.}.
\ei

The following assumes some familiarity with the math behind neural networks, and some comfort with linear algebra and multivariate calculus. Readers who are interested in the results can, with no loss of continuity, merely read the statements of the theorems (skipping the lemmas) and the last paragraphs of each subsection to understand how the smoothness is computed. Our code is also open-sourced, and some readers might find it easier to instead read the implementation.

\subsection{Preliminaries}
\label{sec:prelim}

Our description of {\IT} uses the following notation.

For any learner, $X, y$ will represent a training set containing multiple independent values $X$  and a single goal $y$. In neural networks, we also use $W$ to denote the weights on the links between neurons. 
We also say:
\bi 
\item $m$ is the number of training samples;
\item $n$ is the number of features;
\item $k$ is the number of classes, i.e. possible values for $y$;
\item  $E$ denotes the loss function.
\ei
For a feedforward neural network, $L$ represents the number of layers.
In that learner, at each layer, the following computation is performed:
\begin{align}
    z^{[l]} &= W^{[l]T}a^{[l-1]} + b^{[l]} \label{eq:zl} \\
    a^{[l]} &= g^{[l]}(z^{[l]}) \nonumber \\
    a^{[0]} &= X \nonumber
\end{align}
Here, $b^{[l]}$ is the bias vector at layer $l$, $W^{[l]}$ is the weight matrix at layer $l$, and $g^{[l]}$ is the activation function used at layer $l$ (we assume ReLU \mbox{$g(x) = \max(0, x)$}). The result from applying the activation function is called the activations at layer $l$, and is denoted by $a^{[l]}$.
In this text, $\nabla_W f$ refers to the first gradient of $f$, with respect to $W$. Finally, $\nabla^2_W f$ denotes the second gradient, also called the Hessian. We use $\sup f$ to denote the least upper bound (supremum) of $f$.



An interesting use case for the gradient is the $L-$Lipschitzness, which can be defined as $\sup \lVert \nabla_W f \rVert$\footnote{In fact, it can be shown that a differentiable function $f$ is $L-$Lipschitz iff $\lVert \nabla f \rVert \leq L$} . More strictly, a function $f$ is said to be $L-$Lipschitz if there exists a constant $L$ such that 
\begin{equation}
    \centering
    \frac{\lVert f(y) - f(x) \rVert}{\lVert y - x \rVert} \leq L
\end{equation}
for all $x, y$ in the domain of $f$. The Lipschitz constant gives an idea of the curvature bound of the function under review.

The second gradient (the Hessian) is denoted by $\mathcal{H} = \nabla^2 f$. In the strictest sense, the Hessian is a square tensor of second-order partial derivatives, and is symmetric (by Schwarz's theorem). A continuously differentiable function $f$ is said to be $\beta-$smooth if its \textit{gradient} is $\beta-$Lipschitz, or
\begin{equation}\label{brief}
\frac{\lVert \nabla f(y) - \nabla f(x) \rVert}{\lVert y - x \rVert} \leq \beta
\end{equation}
From the mean-value theorem,   there exists $v$ such that 
\[
    \frac{\lVert \nabla f(y) - \nabla f(x) \rVert}{\lVert y - x \rVert} = \lVert \nabla^2 f(v) \rVert \leq \sup\limits_v \lVert \nabla^2 f(v) \rVert
\]
(Note that this computation   assumes   the loss function is twice-continuously differentiable.)

Hence, $\sup\limits_v \lVert \nabla^2 f(v) \rVert$ is such a $\beta$\footnote{Since all norms are equivalent in finite dimension, the choice of the norm (such as Euclidean or Frobenius) of the Hessian is irrelevant for our application, and different configurations can be ordered by smoothness as long as the choice of norm remains consistent.}. Since $\beta$ is the least such constant, 
$\beta \leq \sup\limits_v \lVert \nabla^2 f(v) \rVert$. 

\changed{3a3.2}{
    Motivated by the unique properties of landscapes generated by SE data, {\IT} solves the following problem:

    $$
    \argmax\limits_{h \in \mathcal{H}} \max\limits_{x \sim \mathcal{X}} \left\lVert\nabla^2_{W^{[L]}} \mathcal{L}(f; x, h) \right\rVert
    $$

    Notice that choosing a specific hyper-parameter configuration and dataset materializes a specific loss function, whose smoothness we can compute. We are instead interested in a bound on the smoothness of all possible landscapes given a hyper-parameter configuration, which the inner $\max$ computes. The outer $\max$ searches across the hyper-parameter space itself in a randomized fashion.
}




\subsection{Smoothness for Feedforward Classifiers} 

This section shows the {\smoothness} computation for different learners. We prove the result for feedforward classifiers here; for other derivations (Naive Bayes, logistic regression), please see the proofs at the end of this paper.  

First, we will prove an auxiliary result that will support our main proof. Readers who are interested in the equation for the smoothness can skip ahead to Theorem \ref{th:ff}.

\begin{lemma}
    \label{lemma:partial:main}
    For a deep learner with ReLU activations in the hidden layers, the last two layers being fully-connected, and a softmax activation at the last layer, 
    \[
        \dfrac{\partial E}{\partial z_j^{[L]}} \leq \frac{k-1}{km} \left( \sum\limits_{i=1}^m [y^{(i)} = j] \right)
    \]
    under the cross-entropy loss.
\end{lemma}

\begin{proof}
    We will use the chain rule, as follows:
    \begin{align}
    	\frac{\partial E}{\partial z^{[L]}_{j}} &= \frac{\partial E}{\partial a^{[L]}_j}\cdot \frac{\partial a^{[L]}_j}{\partial z^{[L]}_j} \label{eq:int:5:main}
    \end{align}
    Consider the Iverson notation version of the general cross-entropy loss:
    $$
    E(\textbf{a}^{[L]}) = -\frac{1}{m} \sum\limits_{i=1}^m \sum\limits_{h=1}^k [y^{(i)} = h] \log a_h^{[L]}
    $$

    Then the first part of \eqref{eq:int:5:main} is trivial to compute:
    \begin{equation}
        \frac{\partial E}{\partial a_j^{[L]}} = -\frac{1}{m} \sum\limits_{i=1}^m \frac{[y^{(i)}=j]}{a_j^{[L]}} \label{eq:mul:2:main}
    \end{equation}
    The second part is computed as follows.
    \begin{align} 
        \frac{\partial a^{[L]}_j}{\partial z^{[L]}_p} &= \frac{\partial}{\partial z^{[L]}_p} \left( \frac{e^{z^{[L]}_j}}{\sum_{l=1}^k e^{z^{[L]}_l}} \right) \nonumber \\ 
        &= \frac{[p = j] e^{z^{[L]}_j}\sum_{l=1}^k e^{z^{[L]}_l } - e^{z^{[L]}_j} \cdot e^{z^{[L]}_p} }{\left( \sum_{l=1}^k e^{z^{[L]}_l } \right)^2} \nonumber \\ 
        &= \frac{[p=j] e^{z^{[L]}_j }}{\sum_{l=1}^k e^{z^{[L]}_l }} - \frac{e^{z^{[L]}_j }}{\sum_{l=1}^k e^{z^{[L]}_l }} \cdot \frac{e^{z^{[L]}_p }}{\sum_{l=1}^k e^{z^{[L]}_l }} \nonumber \\ 
        &= \left([p=j] a^{[L]}_j - a^{[L]}_j a^{[L]}_p \right) \nonumber \\ 
        &= a^{[L]}_j([p=j]-a^{[L]}_p) \label{eq:mul:3:main}
    \end{align}

    Combining \eqref{eq:mul:2:main} and \eqref{eq:mul:3:main} in \eqref{eq:int:5:main} gives
    \begin{equation}
        \frac{\partial E}{\partial z^{[L]}_j} = \frac{1}{m} \left( \sum\limits_{i=1}^m [y^{(i)} = j] \right) \left( \sum\limits_{h=1}^k a^{[L]}_j - [h = j] \right)
        \label{eq:mul:4:main}
    \end{equation}
    As a prudence check, note that $E$ is a scalar and $z^{[L]}_j$ is a \textit{vector} (over all the samples); the right-hand side is also a vector (by vectorizing the Iverson term $[h=j]$). It is easy to show that the limiting case of this is when all softmax values are equal and each $y^{(i)}=p$; using this and $a^{[L]}_j = \frac{1}{k}$ in \eqref{eq:mul:4:main} and combining with \eqref{eq:int:5:main} gives us our desired result:
    \begin{equation}
        \sup\limits_j \norm{ \frac{\partial E}{\partial z^{[L]}_j} } = \frac{k-1}{km} \left( \sum\limits_{i=1}^m [y^{(i)} = j] \right)
    \end{equation}
\end{proof}

\begin{theorem}[Smoothness for feedforward networks]
    For a deep learner with ReLU activations in the hidden layers, the last two layers being fully-connected, and a softmax activation at the last layer, the {\smoothness} of the cross-entropy loss is given by
    \begin{equation}
    \label{eq:beta-ff:main}
        \sup \lVert \nabla^2_W E \rVert \propto \frac{k-1}{km} \sup \frac{\lVert a^{[L-1]} \rVert}{\lVert W^{[L]} \rVert} 
    \end{equation}
    \label{th:ff}
\end{theorem}

\begin{proof}
Lemma \ref{lemma:partial:main} gives us (using the chain rule one step further):
\[
    \nabla_W E \leq \frac{(k-1)}{km} \left( \sum\limits_{i=1}^m [y^{(i)} = j] \right) a_j^{[L-1]}
\]

The $\left( \sum\limits_{i=1}^m [y^{(i)} = j] \right)$ term is constant, so we drop it as a proportionality constant (say, $C$). Therefore,
\[
\begin{aligned}
    \nabla^2_W E &\leq C\frac{k-1}{km} \nabla_W a_j^{[L-1]} & \\
     &= C\frac{k-1}{km} \dfrac{\partial a_j^{[L-1]}}{\partial E} \dfrac{\partial E}{\partial w_{ij}} & \\
     &\leq C\frac{(k-1)^2 a_j^{[L-1]}}{k^2m^2} \dfrac{\partial a_j^{[L-1]}}{\partial E} & \textit{using chain rule on Lemma 1} \\ 
     &= C\frac{(k-1)^2 a_j^{[L-1]}}{k^2m^2} \frac{1}{\dfrac{\partial E}{\partial a_j^{[L-1]}}} & \\
\end{aligned}
\]
Using $z^{[l]} = W^{[l]T}a^{[l-1]} + b^{[l]}$ and Lemma \ref{lemma:partial:main} gives us:
\[
\begin{aligned}
    \dfrac{\partial E}{\partial a_j^{[L-1]}} &= \dfrac{\partial E}{\partial z_j^{[L]}} \dfrac{\partial z_j^{[L]}}{\partial a_j^{[L-1]}} \leq \frac{(k-1)}{km} \lVert W^{[L]} \rVert
\end{aligned}
\]
from which \eqref{eq:beta-ff:main} follows.
\end{proof}

\textit{Note:} It is typical to use a regularization term in addition to the standard loss function above. In the proofs section (at the end of this paper), we show that the addition of the regularizer only adds a constant term to the smoothness, and as such, it can be disregarded.


To report all this maths in   another way, all that is required is to compute the norm of the weight matrix at the last layer. Of the $N_1$ configurations trained for one epoch, we choose the $N_2$ with the highest norms. This is because $k$ and $m$ are constant for a dataset, and so $\frac{k-1}{km}$ is simply a scaling factor.

\subsection{Smoothness for Logistic Regression}
\begin{corollary}[Smoothness of logistic regression]
    \[
        \beta = \frac{k-1}{km} \sup \frac{\lVert X \rVert}{\lVert W \rVert}
    \]
\end{corollary}
\begin{proof}
    Logistic regression can be viewed as a feedforward network with no hidden layers.
Hence, the smoothness for logistic regression follows from Equation~\eqref{eq:beta-ff:main}.
\end{proof}

To say that another way, the smoothness computation for logistic regression is particularly easy: after training the learner, the practitioner simply computes the norm of the input data and the weights after that single epoch of training and uses them to compute the smoothness.



\subsection{Smoothness for Gaussian Naive Bayes}

For the proof of the following, please see the proofs section at the end of this paper.  

\begin{theorem}[Smoothness for Gaussian Naive Bayes]
    For Naive Bayes, the  {\smoothness} is given by
    \[
        \beta = \sup \lVert \boldsymbol\Sigma^{-1} \mathcal{A}G - \frac{1}{2} \boldsymbol\Sigma^{-1} \mathcal{A} \boldsymbol\Sigma^{-1} + G \mathcal{A} \boldsymbol\Sigma^{-1} \rVert
    \]
    where $\boldsymbol\Sigma$ is the covariance matrix, $\mathcal{A}$ is the fourth-order identity tensor, and $G = \boldsymbol \Sigma^{-1} \left( (\boldsymbol x - \boldsymbol \mu) (\boldsymbol x - \boldsymbol \mu)^T + \frac{1}{2} \boldsymbol \Sigma \right) \boldsymbol \Sigma^{-1}$.
    \label{th:nb}
\end{theorem}

Notice that the smoothness is dependent only on the data $X$, means of the classes $\mu_i$, and the covariance matrix $\Sigma$), the latter two of which can also be derived from the dataset.

To say that another way, 
in practice, one would
\begin{itemize}
    \item Extract the means and covariance matrix from the dataset
    \item Compute $G$, which simply requires us to compute the inverse of the covariance matrix and the deviations of each data point from the mean for the corresponding class
    \item Use the intermediaries derived above to compute the smoothness.
\end{itemize}

It is worth noting that this is only true for \textit{Gaussian} Naive Bayes. This is because our proof relies on the fact that Gaussian Naive Bayes is a special case of Gaussian Discriminant Analysis; this is not true for multinomial Naive Bayes.



\subsection{Smoothness for other learners}
\label{sec:finitediff}

While it is possible to use the methodology used in this paper to derive the {\smoothness} for other classical learners, we leave that to future work. Right now we are experimenting with a non-parametric
form of all the above where some recursive cluster method, or a decision tree learner,
divides the data into a tree with small leaves. Then we report smoothness by computing the difference between neighbouring leaves and propagating that change up the tree using a center difference approximation:
\[
    (\nabla^2 E)(x) \approx \frac{E(x+h) - 2E(x) + E(x-h)}{h^2}
\]
where $E$ is an impurity measure and the approximation error is $\mathcal{O}(h^2)$.

Preliminary results are promising but, at this time, we have nothing definitive to report.

\begin{figure*}[!th]
    \centering
    \includegraphics[width=\textwidth]{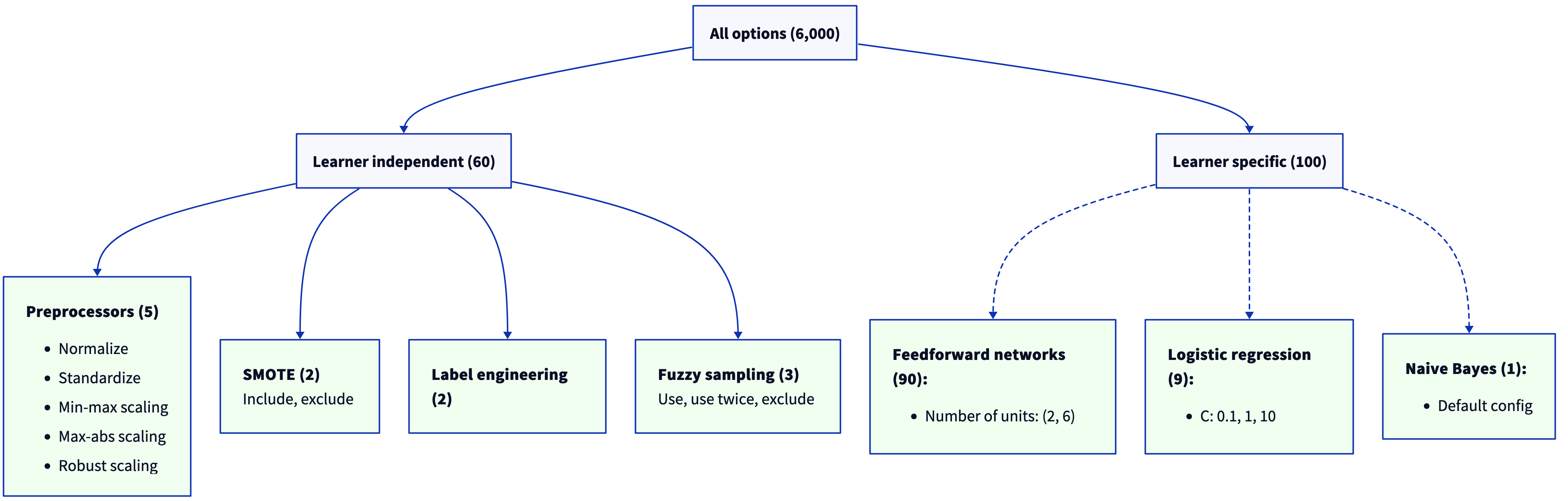}
    \caption{Summary of hyper-parameter space explored in this paper.}
    \label{fig:options-summary}
\end{figure*}

\begin{table}
    \centering
    \caption{Dataset sizes for each of the tasks in this study. Numbers on brackets denote number of SE papers we could find that  have explore these options (in a Google search of "hyper-parameter optimization" and "software engineering" 2014-2024).}
    \label{tab:data}
    { \scriptsize
    \begin{tabular}{llll}
        \toprule
        \multicolumn{3}{c}{Defect prediction \cite{yedida2021value}} \\
        \midrule
        \textbf{Dataset} & \textbf{\#Train (\%defects)} & \textbf{\#Test (\%defects)} & \textbf{\#Features} \\
        \midrule
        camel & 1819 (21) & 442 (64) & \multirow{6}{*}{20} \\
        ivy & 352 (22) & 352 (11) & \\
        log4j & 244 (29) & 205 (92) & \\
        synapse & 379 (20) & 256 (34) & \\
        velocity & 410 (70) & 229 (34) & \\
        xalan & 2411 (38) & 909 (99) & \\
        \midrule
        \multicolumn{3}{c}{Static code warnings \cite{yedida2023find}} \\
        \midrule
        \textbf{Dataset} & \textbf{\#Train} & \textbf{\#Test} & \textbf{\#Features} \\
        \midrule
        ant & 28 & 7 & 329 \\
        cassandra & 13 & 4 & 265 \\
        commons & 17 & 5 & 127 \\
        derby & 366 & 92 & 292 \\
        jmeter & 14 & 4 & 287 \\
        lucene-solr & 24 & 6 & 278 \\
        tomcat & 147 & 37 & 284 \\
        \midrule
        \multicolumn{3}{c}{Issue lifetime prediction \cite{yedida2021old,lee2020continual}} \\
        \midrule
        \textbf{Dataset} & \textbf{\#Train} & \textbf{\#Test} & \textbf{\#Features} \\
        \midrule
        Eclipse & 44,545 & 25,459 & \multirow{3}{*}{12} \\
        Firefox & 44,800 & 25,201 \\
        Chromium & 44,801 & 25,200 \\
        \midrule
        \multicolumn{3}{c}{Bayesmark datasets} \\
        \midrule
        \textbf{Dataset} & \textbf{\#Train} & \textbf{\#Test} & \textbf{\#Features} \\
        \midrule
        breast & 455 & 114 & 30 \\
        digits & 1438 & 359 & 64 \\
        iris & 120 & 30 & 4 \\
        wine & 142 & 36 & 13 \\
        diabetes & 354 & 88 & 10 \\
        \bottomrule
    \end{tabular}
    }
\end{table}

 \section{Experimental Rig}
\label{sec:experiments}

To test {\IT}, we apply the following rig.

\subsection{Data}

Table \ref{tab:data} describes the datasets used in this study. Clearly, our study encompasses a variety of data shapes: both small and large, and both wide and narrow. In particular, for static code warnings, the number of features varies by dataset because of the data collection process from FindBugs; see \cite{kang2022detecting} for a detailed report of the data cleaning process.

\subsection{Learners}
For our learners, we use the three for which we can currently
calculate the smoothness of their loss function; i.e.
feedforward networks, 
Naive Bayes, and logistic regression.

Bayesmark \cite{turner2021bayesian} is a benchmark suite used at the NeurIPS 2020 Black-Box Optimization Challenge. The datasets used in the challenge are available via OpenML\footnote{\url{https://openml.org}}, and are a part of the HPOBench \cite{eggensperger2021hpobench} benchmark. Bayesmark, as a benchmarking tool, uses a normalized scoring system that discourages HPO methods that perform close to random search.

\subsection{Hyper-Parameter Options}

Our hyper-parameter options come from Tables~\ref{tab:hyp:eg}-~\ref{what}, and are summarized in Figure \ref{fig:options-summary}.
We use these options since, as detailed in Table~\ref{what}, all those options have found favor with prior researchers. In Figure \ref{fig:options-summary}, branches with a solid line represent an AND, while dashed lines represent an OR. The learner-independent space has 60 options, and the learners can be configured in 100 ways. This yields 60 * 100 = 6,0000 options in total. For the experiments on the Bayesmark datasets, we use the hyper-parameter space used by Bayesmark, which for feedforward networks optimized by Adam, has a size of 134M options\footnote{\url{https://github.com/uber/bayesmark/blob/master/bayesmark/sklearn_funcs.py#L91-L101}}.
 
Fortunately, many of these options have similar effects.  As show below that within a random selection of 30, it is possible to find settings that significantly improve the performance of an off-the-shelf learner.  


\subsection{Hyper-Parameter Optimizers}
For all our datasets, we compare against BOHB, DEHB, random search, and the current state-of-the-art for those tasks. For {\IT}, we try all the learners whose {\smoothness} value we derived; we use the best-performing learner as the base for random search, BOHB, and DEHB (described above in \S\ref{aboutHPO}). In those trials, we  compared to BOHB after $N_1$ and $N_2$ evaluations, and DEHB after $N_1$ evaluations. This gives both methods the fairest chance.

We briefly mention that although DEHB can be run in a distributed fashion and reduce runtime that way, we instead compare with it in a simpler, single-machine setup. This is since {\IT} is also trivially parallelizable (since all trials in a stage are independent). We argue that this approach maintains fairness when comparing runtimes for all methods.

\subsubsection{Train and Test Sets}
For all experiments, we split the data into train and test partitions, and report results on the test set. We use the splitting policies seen in prior work:
\bi 
\item For defect prediction, it is standard practice to train on data from versions $1\ldots n-1$ and test on version $n$ \cite{agrawal2019dodge};
\item For issue lifetime prediction, we split the data 75/25;
\item For static code analysis, we split the data 80/20. 
\item Bayesmark splits the dataset into 80/20 internally.
\ei

\subsection{Evaluation Metrics}
For each kind of data set, we use evaluation measures seen in prior papers that explored that data:
\bi 
\item  For {\em  issue lifetime prediction},  as done in prior work~\cite{lee2020continual} the class variable was divided into $k\in\{2,3\}$ equal frequency bins. We use the same metrics as \cite{lee2020continual}, i.e. {\em accuracy}. If the 
true negative, false negative, false positive and true positive counts are $a,b,c,d$ 
then
\[
\mathit{accuracy} = (a+d) / (a+b+c+d)
\]
\item 
For {\em defect prediction}, as done in prior work~\cite{wang2016automatically, yedida2021value}, we used the F1
 evaluation metric
that commented not just on how many defects were were found but also, of the predicted defects, how many were true defects:
\[
\begin{array}{rcccl}
\mathrm{recall}        & = & r           & = d / (b+d)\\
\mathrm{precision}     & = & p           & = d / (c+d)\\
\mathrm{harmonic\; mean} & = & \mathit{F1} & = 2rp / (r+p)
\end{array}
\]
\item For {\em static code false alarms}, as done in prior work~\cite{yedida2023find}, we use
recall, precision, and false alarm:
\[
\mathrm{false\ alarm} = c/(a+c)
\]
\ei

To test the generality of {\IT} of a broader range of data sets, we also used the Bayesmark dataset, which was previously used at NeurIPS 2020 in their Black-Box Optimization Challenge. The Bayesmark score scales the model performance based on how well random search does, after normalizing the mean over multiple repeats. The normalized mean score first calculates the performance gap between observations and the global optimum and divides it by the gap between random search and the optimum.

\input{static-code}
\subsection{Statistical Analysis}
\label{sec:stats}

To compare multiple groups of data, we use a combination of the Kruskal-Wallis test (as used in recent SE literature \cite{gopal2022peer, napier2023empirical}) and the Mann-Whitney U-test \cite{mann1947test} at a 5\% significance level. In particular, we first use the Kruskal-Wallis H-test to determine if there is a statistically significant difference between the multiple groups. If there is no significant difference, we report all treatments as indistinguishable; if there is, we run pairwise Mann-Whitney U-tests to find treatments that are not statistically significantly different from the highest-scoring treatment (using a significance level of 0.05). We apply the Benjamini-Hochberg correction to adjust p-values (as endorsed by \citet{farcomeni2008review}) of the Mann-Whitney U-tests.


\section{Results}
\label{sec:results}



\input{issue}
\input{defect}
\input{bbo}

 We ran the experiments for defect prediction and issue lifetime prediction on a university Slurm cluster. For static code analysis, we used a Google Cloud N2 machine with 4vCPUs and 23GB RAM. For the Bayesmark datasets, we ran the experiments on a Google Cloud C3 VM.

 Using those results, we explore three research questions:
\bi 
\item {\bf RQ1:} What are the benefits of smoothing (measured in metrics like F1, recall, etc)?
\item
{\bf RQ2:} What are the costs of smoothing (measured in CPU)?
\item
{\bf RQ3:}   Should we use
different kinds of algorithms for our SE data? Note that this third question addresses the issues raised in our introduction.
\ei
 
\subsection{RQ1: What are the benefits of smoothing?}
 
Table \ref{tab:res:staticcode} shows our results on {\bf static code analysis false alarms}. 
Recall that the task here was to recognize static code warnings that, in
the historical log, were ignored by developers.  In the table, gray denotes ``best results in this column, for that dataset``. Note that {\IT} (and more specifically, {\IT} (FF) has much more gray than anything else.
More specifically:
 all the non-{\IT} methods have far more white cells (i.e. non best) than the {\IT} methods.
(Table \ref{tab:results} shows the exact win/tie counts).

 Table \ref{tab:issue} shows our results on {\bf issue lifetime prediction}. Recall that the task here is to classify a GitHub issue according to how long it will take to resolve. As done in prior work \cite{lee2020continual}, for 3-class classification, the number of days taken to fix bugs is binned into 0 to 1 days, 2 to 5 days, and 5 to 100 days. For 2-class, the number of days taken to fix bugs is binned into less than (or greater than) 4 days for Firefox, less than (or greater than) 5 days for Chromium, and less than (or greater than) 6 days for Eclipse. Note that:
\bi
\item For two class prediction (eclipse(2), firefox(2), chromium(2)) {\IT} does not perform well (as shown in Table \ref{tab:results}, 1 tie and  2 losses). These datasets have the highest smoothness (i.e.,  sharpest minima) over all our SE datasets: the 2-class Chromium, Eclipse, and Firefox datasets had smoothness values of 9.4, 14.9, and 13.8 respectively. As  discussed in RQ3 (Figure \ref{fig:pbetter}), we   expect SMOOTHIE to perform poorly on datasets with such a high smoothness. This explains why {\IT} did not work well for the 2-class binned version of the data.
\item For three-class problems (eclipse(3), firefox(3), chromium(3)) , {\IT} wins in all cases.
\ei

\begin{table}[!t]
\centering
\caption{{\bf RQ1 results (summary):}  from 20 repeats over 24 data sets.}
\label{tab:results}
\scriptsize
\begin{tabular}{llll} 
\toprule
& \multicolumn{3}{c}{Frequency counts of:}\\
\textbf{Datasets}                 & \textbf{wins} & \textbf{ties} & \textbf{losses} \\
\midrule
Static code warnings & 4 & 3 & 0 \\
Issue lifetime prediction (binary) & 0 & 1 & 2 \\
Issue lifetime prediction (multi-class) & 0 & 3 & 0 \\
Defect prediction & 0 & 6 & 0 \\
Bayesmark & 2 & 1 & 2  \\
\bottomrule
\end{tabular}
\end{table}

Table \ref{tab:defect} shows     {\bf defect prediction} results. Recall that this task   imports static code features and outputs a prediction of where bugs might be found.
It turns out that defect prediction is our simplest task  (evidence: all methods achieve the same best ranked results); defect prediction may not be  a useful domain to compare   optimizers.

 
Table \ref{tab:bbo} shows our results for the NeurIPS 2020 Black-
Box Optimization Challenge public dataset. We include these results here just to
study how well {\IT} performs on non-SE data. 

 Table \ref{tab:results} summarizes   the {\bf RQ1} results from  
  Table~\ref{tab:res:staticcode} to Table~\ref{tab:bbo}.
   In that table,  two methods {\bf tie} if, they achieve the same statistical rank (using the methods of \S\ref{sec:stats}). Otherwise, a
{\bf win} is when one measure has a different rank to another {\em and} one measure has better median values (where ``better'' means ``less`` for false alarms and ``greater'' for other metrics).


In summary, from Table~\ref{tab:results}, we say:


\begin{formal}
{\bf RQ1 (answer)}: In the usual case, {\IT} rarely does worse and often
performs better than   the prior state of the art in hyper-parameter optimization.
\end{formal}

\begin{figure} 
    \centering
    \includegraphics[width=\linewidth]{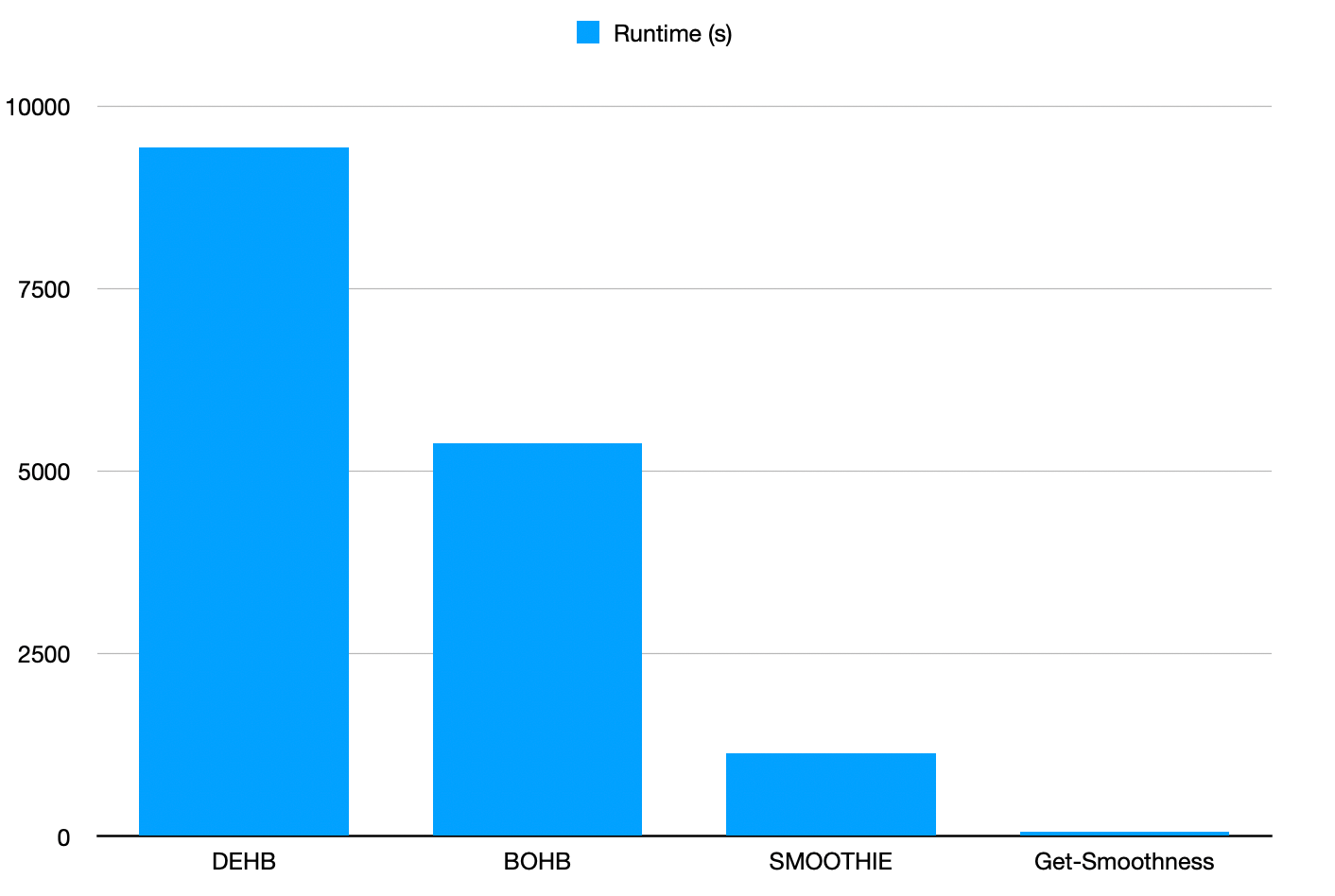}
    \caption{\textbf{RQ2 Results:} Runtime   on   chromium-2   dataset.}
    \label{fig:runtime}
\end{figure}

\subsection{RQ2: What are the costs of smoothing?}
Figure \ref{fig:runtime} shows the runtimes for our method. We take care to distinguish the overheads associated with finding smoothness, versus the time required for another process to use smoothness to make a recommendation about what are good hyperparameter options. All algorithms were measured on a Google Cloud n1-standard-8 VM with a P4 GPU.

As shown in Figure \ref{fig:runtime}, the time required to compute {\tt Get-Smoothness} smoothness is negligible.  Those results are then used by {\IT} for its decision making (and here, {\IT} repeats its processing $N_1+N_2=30+5$ times).  Even with those repeats, {\IT} is 4.7 times faster than BOHB and 8.3 times faster than DEHB.

Figure \ref{fig:runtime} suggests that further optimizations might be possible if {\IT} could
reduce how many repeats its executes. This is a matter for further work, in another paper.

Overall, we say: 
\begin{formal}
{\bf RQ2 (answer)}:   {\IT} works as well as the prior state of the art,  and does so
four times faster.  
\end{formal}

\begin{figure}[!t]
    \centering
        \includegraphics[width=\linewidth]{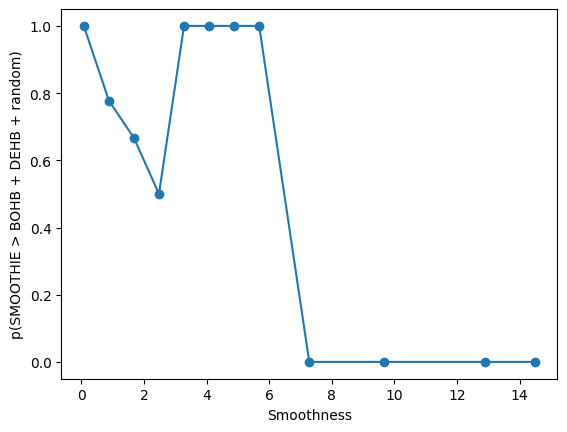}
        \caption{Probability of SMOOTHIE outperforming (or matching) BOHB, DEHB, and random as a function of smoothness}
        \label{fig:pbetter}  
\end{figure}

\subsection{RQ3: Should  we use
different kinds of algorithms for our SE data?}

Turning now to the  problem sketched   in our introduction,   recall
Figure~\ref{fig:beta-dist} where the mean smoothness of  SE, AI data was   2,5, respectively. 
Figure~\ref{fig:pbetter}, shows the  reports the probability of {\IT} outperforming BOHB's F1 score, measured as the fraction of datasets where
{\IT} outperforms BOHB\footnote{More precisely, $\forall i > 0, f(x_i) = \frac{\mathbbm{1}\{ (x_{i-1} \leq x \leq x_i) \land (\text{SMOOTHIE} > \text{BOHB}) \}}{\mathbbm{1}\{ (x_{i-1} \leq x \leq x_i) \}}$}

We note that, in those results,  different algorithms are the best choice
 for low and high smoothness values:
 \bi 
 \item {\IT} usually wins at the  low
 values typically associated with SE data;
\item 
While BOHB, DEHB usually wins at the higher values, seen only in the non-SE data.
\ei

\noindent 
Hence we say:
\begin{formal}
{\bf RQ3 (answer)}: SE data can be different to other kinds of data; and
those differences mean that we should use
different kinds of algorithms for our data.
\end{formal}

\subsection{Sensitivity Analysis}

We also conduct a sensitivity analysis over the parameters $N_1$ and $N_2$ on the static code analysis datasets (excluding maven, due to its smaller size). Figure \ref{fig:sensitivity} shows these results.

On the x-axis is the number of initial evaluations ($N_1$ in this paper) performed, and each line corresponds to one choice of full runs performed ($N_2$ in this paper). We start with a minimum of $N_1 = 20$, and go up to $N_1 = 75$, and test with $N_2 \in \{ 5, 10, 15 \}$. The y-axis shows the \textit{mean normalized regret} over evaluations, where \textit{lower is better}. This normalized regret is calculated as

$$
\text{Normalized regret} = 1-\frac{p_t}{p_{20}}
$$

where $p_t$ is the performance (measured as the difference between recall and false positive rate). We note that as $N_1$ increases, the normalized regret \textit{increases}. This is because $N_2$ stays fixed for each line, so the highest smoothness changes as $N_1$ increases. Further, we note that up to $N_1 \approx 45$, the normalized regret for $N_2 = 5$ is low (and lower than $N_2 = 15$). Past this point, the normalized regret is higher. However, it is worth noting that moving from $N_2 = 5$ to $N_2 = 10, 15$ incurs a 2-3x runtime penalty.

This figure shows that $N_1 = 30, N_2 = 5$ are reasonable default values for these parameters.

\section{Discussion}
\label{sec:discussion}

\subsection{ Why  Does {\IT} Work? }

\label{sec:why-smoothie}

Figure~\ref{fig:pbetter}  raises three questions:
\bi 
\item Why does BOHB fail for small $\beta$-smoothness values (i.e. the low $x$ values of     Figure~\ref{fig:pbetter})?
\item Why do those same flatter regions help {\IT}?
\item Why does {\IT} fail for for very bumpy regions (i.e. the high $x$ values of     Figure~\ref{fig:pbetter})?

\ei

\begin{figure}
    \centering
    \includegraphics[width=\linewidth]{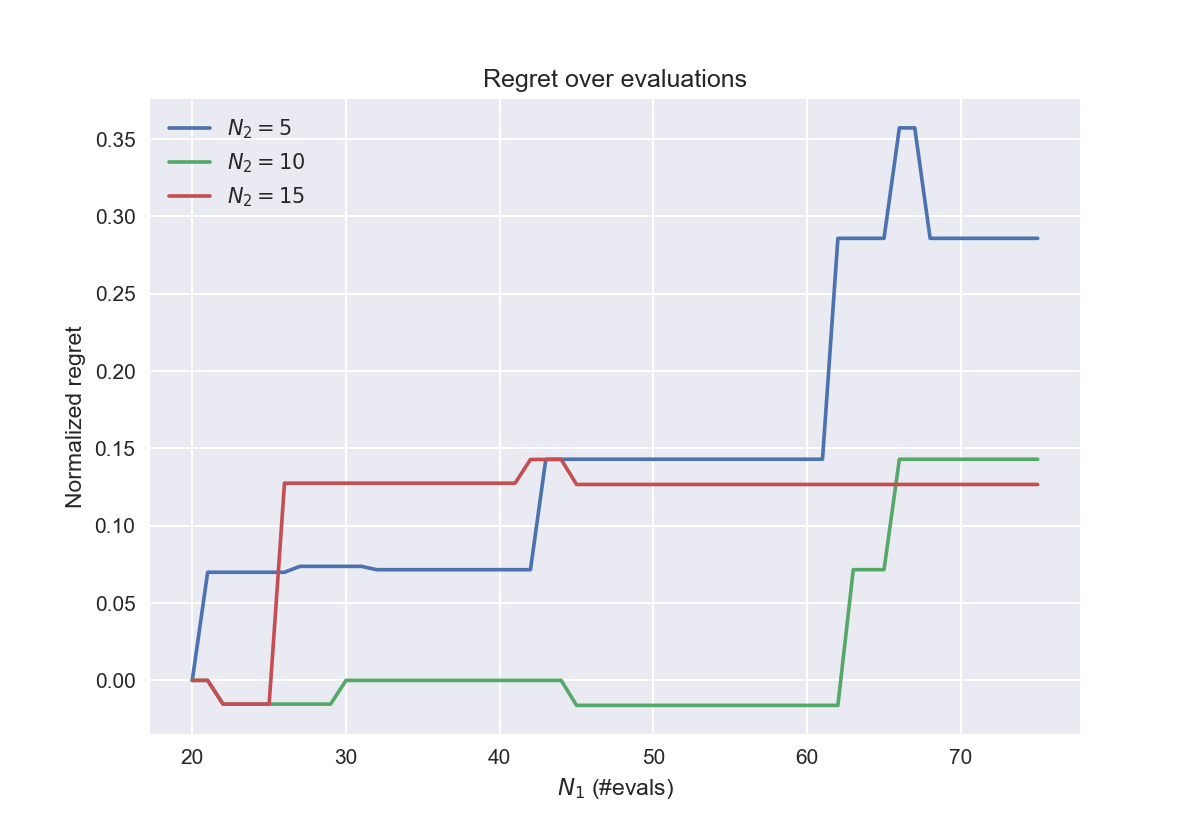}
    \caption{Sensitivity analysis for the static code analysis datasets, excluding maven.}
    \label{fig:sensitivity}
\end{figure}

We offer the hypothesis that algorithms like BOHB 
need a bumpy landscape to distinguish between different parts of the data, in order to decide what regions to explore, so BOHB prefers a bumpier space.

But if  BOHB   fails  in flat regions, why would {\IT} succeed in the same flatter regions?
 It is well-established that learning algorithms (such as SGD or Adam) have a trade-off between {\em training rates} and {\em generalization} \cite{feng2023activity, jastrzkebski2017three, li2022happens, cowsik2022flatter, keskar2017improving, wilson2017marginal, luo2018adaptive, zhou2020towards}.
 Steep landscape decrease training rates since at any point there is an obvious
 better place to jump to. But flatness helps generalization
 (i.e. scores on hold-out test data).
Many theorists~\cite{keskar2016large,jastrzkebski2017three,hochreiter1997flat,dziugaite2017computing,kaddour2022flat,wu2023implicit,jiang2019fantastic,neyshabur2017exploring,hinton1993keeping,chaudhari2019entropy,seong2018towards} comment  that the   flatter the  loss landscape, the more generalizable the conclusion (since for any given conclusion, little changes in   adjacent regions). To see this, consider the what happens when the loss landscape is {\em not } flat:  what is learned in one region might change significantly in the adjacent region.

In short, we believe that BOHB does not explicitly bias for generalization, only for test performance,
and {\IT} gets higher scores on test data by 
better exploiting   generalization over flatter terrains.

As to
why {\IT} fails for bumpy spaces, we say its  design trade-offs do not favor such a space, for the following reason.
 We conjecture that  {\IT}  flattens the loss function {\em by some ratio} of, say, 50\%. If this were true then when 
 $\beta$-smoothness has a large value  of, say, 5.6
 (see  Figure~\ref{fig:illust}.d),
{\IT} can only flatten the landscape to .5*5.6=2.8 (see  Figure~\ref{fig:illust}.c, which is still quite steep). On the other hand, if we start near the mean smoothness of our
SE data (see~Figure~\ref{fig:illust}.b) then a 50\% flattening takes us into a quite flat region (see Figure~\ref{fig:illust}.a) where 
models generalize better.

\begin{figure}[!t]
    \begin{subfigure}{.5\linewidth}
        \begin{center}\includegraphics[height=.7\linewidth]{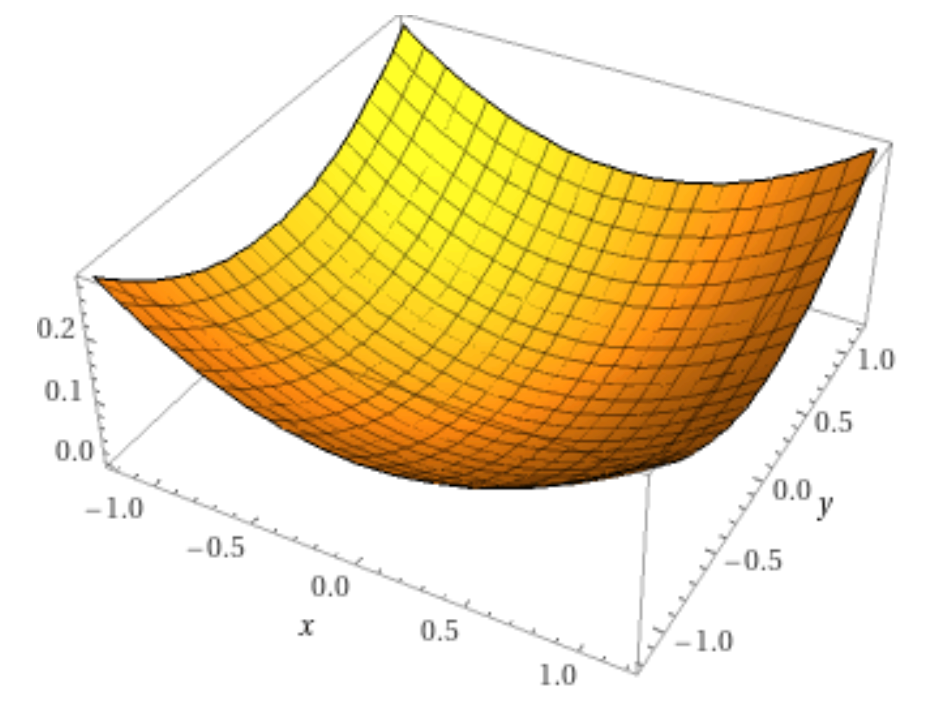}\end{center}
        \caption{$z = 0.3 (x^2 + y^2), \beta = 0.85$}
    \end{subfigure}
    \begin{subfigure}{.5\linewidth}
          \begin{center}\includegraphics[height=.7\linewidth]{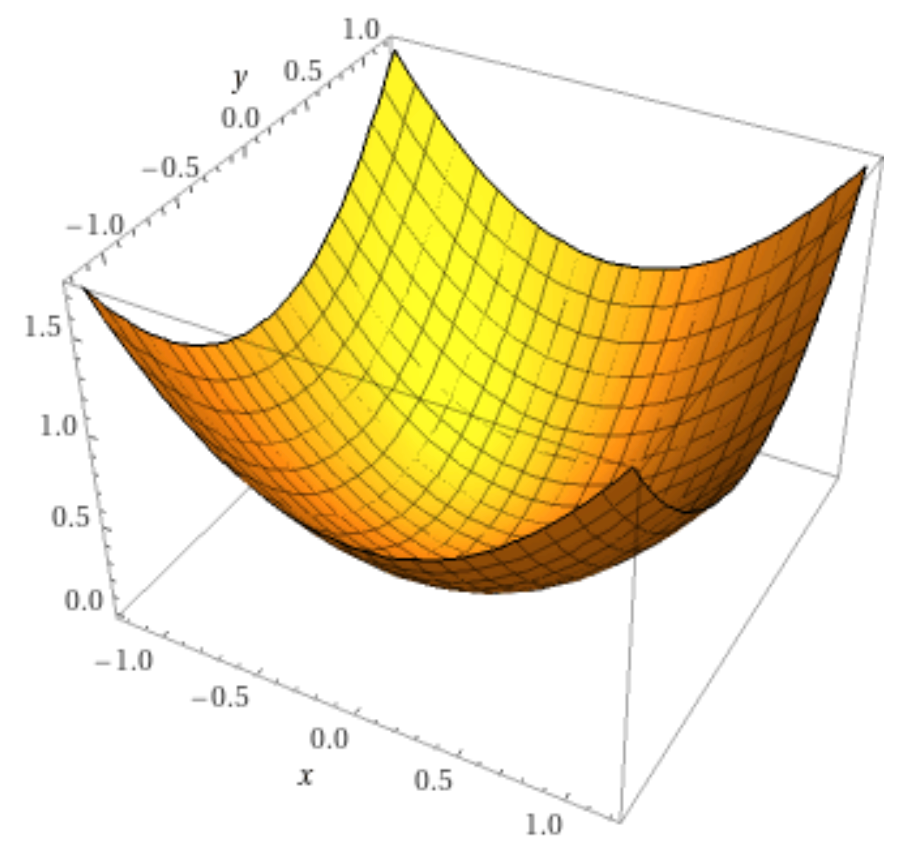}\end{center}
        \caption{$z = 0.6 (x^2 + y^2), \beta = 1.7$}
    \end{subfigure}
    \\
    \begin{subfigure}{.5\linewidth}
          \begin{center}\includegraphics[height=.7\linewidth]{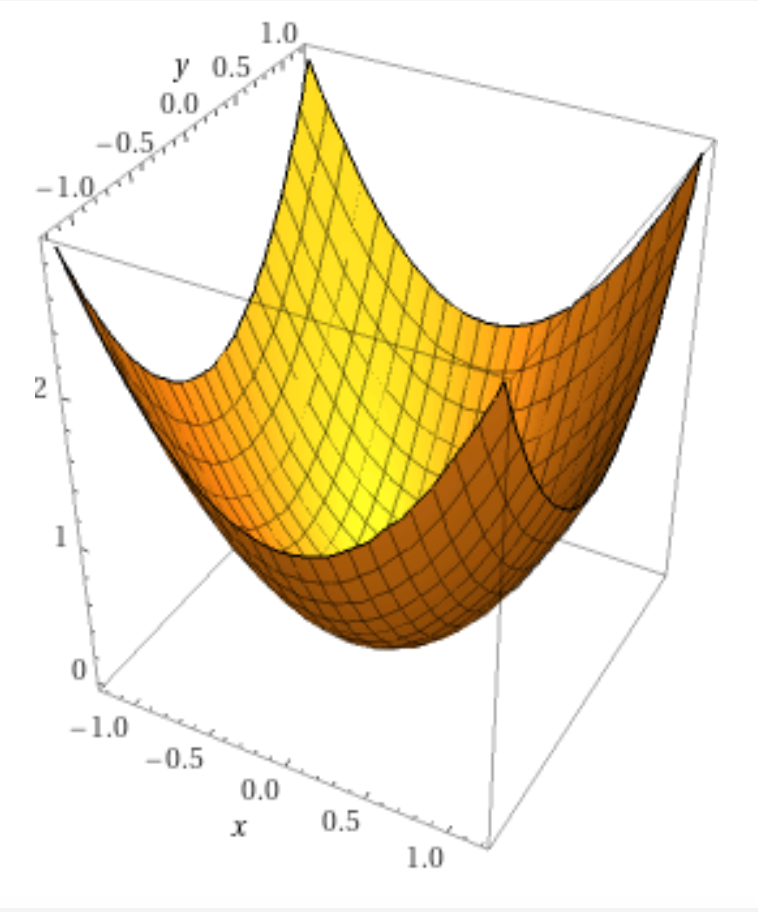}\end{center}
        \caption{$z = x^2 + y^2, \beta = 2.8$}
    \end{subfigure}
    \begin{subfigure}{.5\linewidth}
           \begin{center}\includegraphics[height=.7\linewidth]{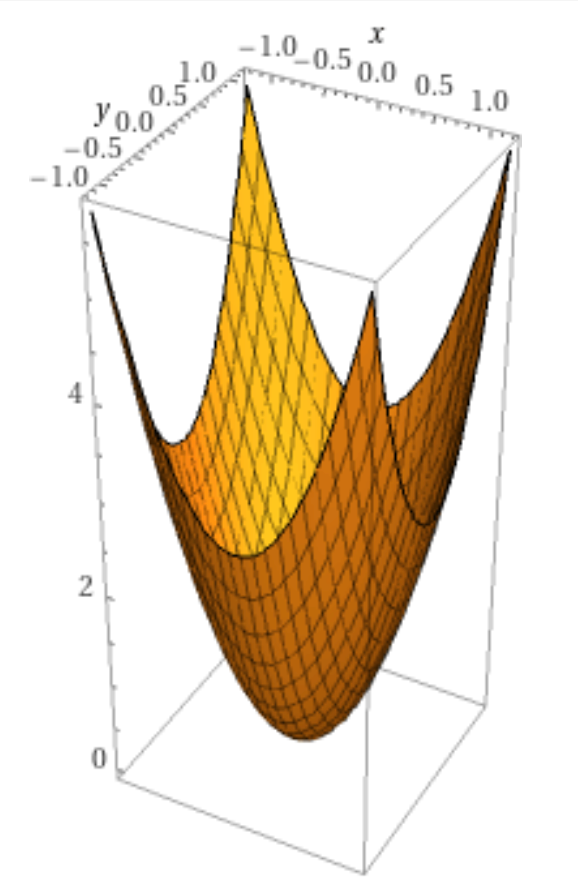}\end{center}
        \caption{$z = 2 (x^2 + y^2), \beta = 5.6$}
    \end{subfigure}
    \caption{An illustration of sharp minima forming with higher smoothness. Figure \ref{fig:illust}d is the upper limit of the sharpness that {\IT} can handle (see Figure \ref{fig:pbetter}).}  
    \label{fig:illust}
\end{figure}

\subsection{Threats to Validity}

In this section, we discuss various biases that might affect the final results, and how our experimental setup attempts to mitigate them.

\textbf{Sampling bias.} Sampling bias is an issue for experimental papers that perform evaluation on a limited set of data points. We believe this is mitigated by testing our method on 24 datasets against a state-of-the-art optimizer. Our code is open-source, and we invite other researchers to extend our results to more datasets.

\textbf{Evaluation bias.} Evaluation bias threatens papers that use a limited number of metrics to evaluate performance. It is important to ensure all comparisons are fair and provide a broad view of the performance. By using the same metrics as prior work to compare the methods in this paper, we demonstrate that our method can achieve state-of-the-art results across a variety of tasks and metrics.

\textbf{Learner bias.} Another source of bias is the learners used in the study. Currently, there are many algorithms for classification (which this paper focuses on). We developed our theory for three learners and tested them across all our datasets in an effort to mitigate this bias. In future work, we will implement the finite difference approximation (cf. Section \ref{sec:finitediff}) which open up the theory to arbitrary algorithms.

\textbf{Comparison bias.} There are many approaches for the tasks we studied in this paper, and it would be difficult to compare to all of them. As such, we decided to compare our work to the current state-of-the-art, as well as a state-of-the-art hyper-parameter optimization algorithm. This ensures that our results are compared against the best among the algorithms in the literature.

\textbf{Order bias.} Order bias has to do with the order of samples affecting performance. For defect prediction, we ensured that the test set is from a newer version than the training set (since that data comes with time stamps).   For all other datasets, we mitigate this bias by randomly splitting the dataset into train and test sets, and repeating our results 20 times.

\textbf{Statistical validity.} Especially with stochastic algorithms, it is important to ensure that any performance differences cannot be explained by noise. We repeat all our experiments 20 times and use adjusted p-values from a widely-used statistical test to declare wins, ties, and losses.

\section{Limitations}
\label{sec:limitations}

The major limitation of this work is the choice of measurement of smoothness, and the use of {\smoothness} as a heuristic. For example, having derived the Hessian, the condition number is an alternate measure of the ``nice nature'' of the landscape \cite{boyd2004convex}. 

Another clear limitation is that this approach is limited to learners for which a {\smoothness} can be derived. Especially in the case where the Hessian may not exist (e.g., if the loss function is not twice-differentiable such as the mean absolute error loss), this approach fails. Nevertheless, we have shown successful application in a classical learner that is not gradient-based, Naive Bayes. In that case, we used the negative log-likelihood as a surrogate for the loss function. 

\changed{3a1.1}{A further    limitation of this work is that because the smoothness computation is different for each type of learner, this can potentially make the implementation less trivial, especially if the practitioner is interested in trying multiple learners. Our reference implementation uses function overloading to choose the correct equations given the type of the classifier object. The white-box nature of our method means that practitioners need to select a list of classifiers \textit{a priori} and implement the corresponding smoothness equations.} 

Finally, we reiterate that {\IT} performs poorly on datasets with high smoothness, as evidenced by our results on the 2-class issue close time prediction results.

\section{Future Work}
\label{sec:futurework}

\changed{Ea3.1}{
This paper showed that for non-neural learners, we can use the negative log-likelihood as a surrogate for the loss (using Gaussian Naive Bayes). Our future work will expand this to other learners, such as tree-based learners, which have strong performance on tabular data \cite{grinsztajn2022tree}. These algorithms are based on maximizing the \textit{information gain} (synonymous with the KL divergence) from splitting nodes (which hold subsets of the dataset) on different features. Suppose node $N$ is split into nodes $\{ N_i \}_{i=1\ldots m}$. Then, the information gain is:

\[
    IG(N) = H(N) - \sum\limits_{i=1}^m \frac{\lvert N_i \rvert}{\lvert N \rvert} H(N_i)
\]

where $H(\cdot)$ is an \textit{entropy measure}. Most tree-based algorithms use either the \textit{Gini index} or the \textit{Shannon entropy} for $H$, but it is worth noting that both are special cases of Tsallis' entropy \cite{tsallis1988possible}. We can then use the negative log-likelihood, which is given by
\[
\ell := \text{NLL} = -2 \lvert E \rvert \cdot IG
\]
which follows since the log-likelihood (also called the G-statistic) is a scaled version of the KL divergence. However, since decision trees are non-parametric, we cannot use the Hessian-based smoothness measure; instead, we need to use a center-difference approximation for the smoothness:
\[
\beta(x) = \frac{\ell(x - h) + \ell(x + h) - 2\ell(x)}{h^2}
\]
where $x$ is some point, and $h$ is a small constant. Different methods of estimating the smoothness of the overall learner based on the above can be used, such as a Monte Carlo estimation.
}

Another avenue of future work would be replacing the random sampling at the beginning of the approach (Algorithm \ref{alg:smoothie}, Line 1) with a smarter selection approach. For example, Bayesian optimization could be used in a loop to search for configurations with high {\smoothness} values. Alternatively, one could use multiple metrics for smoothness, such as the condition number of the Hessian, and use a multi-objective optimization algorithm.

\section{Conclusion}
One of the black arts of software analytics is how
to configure our data mining systems.
Hyper-parameter optimization
is often implemented as a blind search through a large space of options (like randomly hitting a broken car with a hammer, until the car starts working).

This paper offers a  novel theory that can better guide
the design hyper-parameter optimizers:
\begin{quote}{\em Hyper-parameter optimization for SE is best viewed as a
“smoothing” function for the decision landscape since, after optimization, the loss function gets “smoother’
(i.e. on average, fewer changes in any local region).}\end{quote}
This theory predicts that methods that   maximize the smoothness should also improve the hyper-parameter optimization. This prediction was checked on  on 24 datasets across a variety of tasks. As predicted, {\IT}, achieves state-of-the-art performance and outperforms other popular HPO methods.
One reason to endorse this theory is  the time required to compute smoothness is negligible. 

 We show that a hyper-parameter optimizer called {\IT}, which is based on this theory:
 \bi 
 \item Works
 well for flatter landscapes;
 \item But, compared to standard AI tools like 
 BOHB, fails for much bumpier spaces.
 \item On the other hand the standard AI
 tools fail for the flatter spaces.
 \ei This is a significant observation since, as we shown in Figure~\ref{fig:beta-dist}, SE data can   be much flatter than the data usually processed by the AI tools. 

 Based on that results, we conclude that SE data can be different to other kinds of data; and
those differences means that, sometimes, we should use
different kinds of algorithms for our data. Hence it is {
\underline{\bf unwise}
to always use these AI algorithms  
in their default “off-the-shelf” configuration.

\section*{Acknowledgements}

This work was partially funded by an NSF grant \#1908762.

\small
\bibliographystyle{plainnat}
\bibliography{cite}

\input{biography}

\include{proofs}

\end{document}

%% file: static-code.tex
\begin{table} 

\scriptsize
\begin{tabular}{lllrrr}
\toprule
\multicolumn{1}{l}{\textbf{Dataset}} & \textbf{Method (Learner)} & \textbf{\#evals} & \multicolumn{1}{l}{\textbf{pd}} & \multicolumn{1}{l}{\textbf{pf}} & \multicolumn{1}{l}{\textbf{precision}} \\
\midrule
\multirow{8}{*}{ant}                                    & Baseline (SVM)   & & 0 & 0.25 & 0 \\
\cmidrule{2-6}
                                                        & {\IT} (FF) & 10        & \gray{0.65}                   & \gray{0}                      & \gray{0.8}                           \\
                                                        & {\IT} (NB) & 10 & \gray{0.6} & \gray{0} & \gray{0.8} \\
                                                        & {\IT} (LR) & 10 & \gray{1} & \gray{0} & \gray{1} \\
\cmidrule{2-6}
                                                        & Random & 10         & 0.2 & 0.5 & 0.5 \\
                                                        & BOHB (FF) & 10           & \gray{1}                      & \gray{0.2}                    & 0.5                                    \\
                                                        & BOHB (FF) & 50           & \gray{0.6}                    & 1         & 0.6 \\
                                                        & DEHB (FF) & 50           & \gray{1}                      & 0.8    & 0.3 \\ 
\midrule
\multirow{8}{*}{cassandra}                              & Baseline (SVM) &  & 0.33 & 1 & 0.5 \\ 
\cmidrule{2-6}
                                                        & {\IT} (FF) & 10       & \gray{0.83}                   & \gray{0}                      & \gray{0.75}                          \\
                                                        & {\IT} (NB) & 10 & \gray{1} & 0.5 & \gray{0.83} \\
                                                        & {\IT} (LR) & 10 & \gray{1} & 0.67 & 0.33 \\
\cmidrule{2-6}
                                                         & Random & 10         & 0 & 0.66 & 0 \\
                                                        & BOHB (FF)   & 10         & \gray{1}                      & 0.67                            & 0.25                                   \\
                                                        & BOHB (FF)  & 50 & 0 & \gray{0} & 0 \\
                                                        & DEHB (FF)  & 50 & \gray{1} & 0.67 & \gray{0.75} \\
\midrule
\multirow{8}{*}{commons}                                & Baseline (SVM) &  & \gray{1} & 1 & \gray{0.6} \\
\cmidrule{2-6}
                                                        & {\IT} (FF)  & 10      & \gray{1}                      & \gray{0.25}                   & \gray{0.5}                           \\
                                                        & {\IT} (NB) & 10 & 0 & \gray{0} & 0 \\
                                                        & {\IT} (LR) & 10 & \gray{1} & 1 & \gray{0.8} \\
\cmidrule{2-6}
                                                     & Random & 10         & 0 & \gray{0} & 0 \\
                                                    & BOHB (FF)   & 10         & \gray{0.88}                   & 1                               & \gray{0.65}                          \\
                                                    & BOHB (FF) & 50 & \gray{1} & 1 & 0.2 \\
                                                    & DEHB (FF) & 50 & \gray{1} & 0.5 & 0.33 \\
                                                    
\midrule
\multirow{8}{*}{derby}                                  & Baseline (SVM) &  & 0.61 & 0.67 & 0.39 \\
\cmidrule{2-6}
& {\IT} (FF) & 10       & 0.74                   & 0.46                   & 0.44                                   \\
                                                        & {\IT} (NB) & 10 & \gray{0.90} & 0.60 & 0.4 \\
                                                        & {\IT} (LR) & 10 & 0.59 & \gray{0.06} & \gray{0.8} \\
\cmidrule{2-6}
                                                        & Random & 10         & \gray{1.0} & 1.0 & \gray{0.68} \\
                                                        & BOHB (FF) & 10           & \gray{0.93}                   & 0.84                   & \gray{0.68}                          \\
                                                        & BOHB (FF) & 50 & \gray{1} & 1 & \gray{0.68} \\
                                                        & DEHB (FF) & 50 & \gray{0.9} & 0.68 & 0.38 \\
\midrule
\multirow{8}{*}{jmeter}                                 & Baseline (SVM) &  & 0 & \gray{0} & 0 \\ 
\cmidrule{2-6}
                                                        & {\IT} (FF) & 10       & \gray{1}                      & \gray{0}                      & \gray{1}                             \\
                                                        & {\IT} (NB) & 10 & \gray{1} & 0.33 & 0.5 \\
                                                        & {\IT} (LR) & 10 & 0.67 & \gray{0} & 0.63 \\
\cmidrule{2-6}
                                                        & Random & 10         & 0 & 1 & 0 \\
                                                         & BOHB (FF)  & 10          & 0                               & \gray{0}                      & 0                                      \\
                                                         & BOHB (FF) & 50 & 0 & \gray{0} & 0 \\
                                                         & DEHB (FF) & 50 & \gray{1} & 0.33 & 0.5 \\
\midrule
\multirow{8}{*}{lucene-solr}                            & Baseline (SVM) &  & 0 & 1 & 0 \\ 
\cmidrule{2-6}
                                                        & {\IT} (FF) & 10       & \gray{0.5}                    & \gray{0.25}                   & \gray{0.5}                           \\
                                                        & {\IT} (NB) & 10 & \gray{1} & \gray{0.5} & \gray{0.5} \\
                                                        & {\IT} (LR) & 10 & \gray{1} & \gray{0.25} & \gray{0.67} \\
\cmidrule{2-6}
                                                        & Random & 10         & \gray{1} & \gray{0.25} & \gray{0.67} \\
                                                        & BOHB (FF)  & 10          & \gray{0.75}                   & \gray{0.5}                    & \gray{0.33}                          \\
                                                        & BOHB (FF) & 50 & 0 & 0 & 0 \\
                                                        & DEHB (FF) & 50 & \gray{1} & \gray{0.5} & \gray{0.5} \\
\midrule
\multirow{8}{*}{tomcat}                                 & Baseline (SVM) &  & 0.73 & 0.77 & 0.39 \\
\cmidrule{2-6}
                                                        & {\IT} (FF) & 10       & 0.14                            & \gray{0.09}                   & \gray{0.47}                          \\
                                                        & {\IT} (NB) & 10 & 0.43 & \gray{0.13} & \gray{0.63} \\
                                                        & {\IT} (LR) & 10 & 0.71 & 0.3 & \gray{0.61} \\
\cmidrule{2-6}
                                                        & Random & 10         & 0.43 & 0.5 & \gray{0.59} \\
                                                        & BOHB (FF)  & 10          & 0.7                             & 0.48                            & \gray{0.6}     \\
                                                        & BOHB (FF) & 50 & \gray{1} & 1 & \gray{0.62} \\
                                                        & DEHB (FF) & 50 & \gray{1} & 1 & 0.38 \\
\bottomrule
\end{tabular}
\caption{{\bf RQ1 results:}  Static code warnings results. Anything that ties with the statistically best results are shown in \fcolorbox{\graycolor}{\graycolor}{\textbf{bold}} (and such ties
are determined using the statistics of \S\ref{sec:stats}).
For Algorithm~\ref{alg:smoothie}, we use $N_1 = 50, N_2 = 10$. FF = Feedforward network. NB = Naive Bayes. LR = Logistic Regression. 
}
\label{tab:res:staticcode}
\end{table}

%% file: issue.tex
\begin{table}[!b]
\centering

\scriptsize
\begin{tabular}{p{1.3cm}l l@{~}rr}
\toprule
\textbf{Dataset}                 & \textbf{Learner}   & \textbf{HPO method} & \multicolumn{1}{l}{\textbf{\#evals}} & \multicolumn{1}{l}{\textbf{Accuracy}} \\
\midrule
\multirow{8}{*}{eclipse (2)}  & \multicolumn{3}{l}{Baseline}                                       & 61                                    \\
\cmidrule{2-5}
                                 & Feedforward        & {\IT}          & 5                                    & 69.3                                  \\
                                 & Naive Bayes        & {\IT}          & 5                                    & 63                                    \\
                                 & Logistic           & {\IT}          & 5                                    & 62.5                                  \\
\cmidrule{2-5}
                                 & Best (Feedforward) & Random              & 5                                    & \gray{80.2}                         \\
                                 & Best (Feedforward) & BOHB                & 5                                    & 45.8                                  \\
                                 & Best (Feedforward) & BOHB                & 30                                   & 45.8                                  \\
                                 & Best (Feedforward) & DEHB                & 30                                   & 68.7 \\
\midrule
\multirow{8}{*}{firefox (2)}  & \multicolumn{3}{l}{Baseline}                                       & 67                                    \\
\cmidrule{2-5}
                                 & Feedforward        & {\IT}          & 5                                    & \gray{69}                           \\
                                 & Naive Bayes        & {\IT}          & 5                                    & 66.6                                  \\
                                 & Logistic           & {\IT}          & 5                                    & 66.4                                  \\
\cmidrule{2-5}
                                 & Best (Feedforward) & Random              & 5                                    & \gray{76.5}                         \\
                                 & Best (Feedforward) & BOHB                & 5                                    & 36                                    \\
                                 & Best (Feedforward) & BOHB                & 30                                   & 35.9                                  \\
                                 & Best (Feedforward) & DEHB                & 30                                   & \gray{72.4} \\
\midrule
\multirow{8}{*}{chromium (2)} & \multicolumn{3}{l}{Baseline}                                       & 63                                    \\
\cmidrule{2-5}
                                 & Feedforward        & {\IT}          & 5                                    & 62.2                                  \\
                                 & Naive Bayes        & {\IT}          & 5                                    & 63                                    \\
                                 & Logistic           & {\IT}          & 5                                    & 63.1                                  \\
\cmidrule{2-5}
                                 & Best (Logistic)    & Random              & 5                                    & \gray{67.1}                         \\
                                 & Best (Logistic)    & BOHB                & 5                                    & 44.5                                  \\
                                 & Best (Logistic)    & BOHB                & 30                                   & 44.4                                  \\
                                 & Best (Logistic)    & DEHB                & 30                                   & \gray{66.8} \\
\midrule
\multirow{8}{*}{eclipse (3)}  & \multicolumn{3}{l}{Baseline}                                       & 44                                    \\
\cmidrule{2-5}
                                 & Feedforward        & {\IT}          & 5                                    & \gray{71.2}                         \\
                                 & Naive Bayes        & {\IT}          & 5                                    & 51.6                                  \\
                                 & Logistic           & {\IT}          & 5                                    & 51.2                                  \\
\cmidrule{2-5}
                                 & Best (Feedforward) & Random              & 5                                    & 36.3                                 \\
                                 & Best (Feedforward) & BOHB                & 5                                    & 29.2                                  \\
                                 & Best (Feedforward) & BOHB                & 30                                   & 46.2                                  \\
                                 & Best (Feedforward) & DEHB                & 30                                   & \gray{71.2} \\
\midrule
\multirow{8}{*}{firefox (3)}  & \multicolumn{3}{l}{Baseline}                                       & 44                                    \\
\cmidrule{2-5}
                                 & Feedforward        & {\IT}          & 5                                    & \gray{69.2}                         \\
                                 & Naive Bayes        & {\IT}          & 5                                    & 49.5                                  \\
                                 & Logistic           & {\IT}          & 5                                    & 49.3                                  \\
\cmidrule{2-5}
                                 & Best (Feedforward) & Random              & 5                                    & 46.1                                  \\
                                 & Best (Feedforward) & BOHB                & 5                                    & 32.9                                  \\
                                 & Best (Feedforward) & BOHB                & 30                                   & 44.4                                  \\
                                 & Best (Feedforward) & DEHB                & 30                                   & \gray{70.5} \\
\midrule
\multirow{8}{*}{chromium (3)} & \multicolumn{3}{l}{Baseline}                                       & 43                                    \\
\cmidrule{2-5}
                                 & Feedforward        & {\IT}          & 5                                    & \gray{71.9}                        \\
                                 & Naive Bayes        & {\IT}          & 5                                    & 49.1                                  \\
                                 & Logistic           & {\IT}          & 5                                    & 49.8                                  \\
\cmidrule{2-5}
                                 & Best (Feedforward) & Random              & 5                                    & 53.5                                  \\
                                 & Best (Feedforward) & BOHB                & 5                                    & 27.5                                  \\
                                 & Best (Feedforward) & BOHB                & 30                                   & 44.6                                  \\
                                 & Best (Feedforward) & DEHB                & 30                                   & \gray{72.9} \\
\bottomrule
\end{tabular}
\caption{{\bf RQ1 results:} Issue lifetime prediction results.
For Algorithm~\ref{alg:smoothie}, we use $N_1,N_2 =30, 5$.
Anything that ties with the statistically best results are shown in \fcolorbox{\graycolor}{\graycolor}{\textbf{bold}} (and such ties
are determined using the statistics of \S\ref{sec:stats}).
Our baseline is DeepTriage \cite{mani2019deeptriage}. In the dataset column, the number of classes the data is binned into is in parentheses.
}
\label{tab:issue}
\end{table}

%% file: defect.tex
\begin{table}[!b]
\centering

\scriptsize
\begin{tabular}{lllrr}
\toprule
\textbf{Dataset}          & \textbf{Learner} & \textbf{HPO method} & \multicolumn{1}{l}{\textbf{\#evals}} & \multicolumn{1}{l}{\textbf{F1}} \\ 
\midrule
\multirow{8}{*}{camel}    & \multicolumn{3}{l}{Baseline (GHOST)}                                          & \gray{90.9}                  \\
\cmidrule{2-5}
                          & Feedforward      & {\IT}          & 5                                    & \gray{90.9}                   \\
                          & Naive Bayes      & {\IT}          & 5                                    & \gray{90.9}                   \\
                          & Logistic         & {\IT}          & 5                                    & \gray{90.9}                   \\
\cmidrule{2-5}
                          & Best (NB)        & Random              & 5                                    & \gray{90.9}               \\
                          & Best (NB)        & BOHB                & 5                                    & \gray{90.9}               \\
                          & Best (NB)        & BOHB                & 30                                   & \gray{90.9}               \\
                          & Best            & DEHB                  & 30                                  & 42.5 \\
\midrule
\multirow{8}{*}{ivy}      & \multicolumn{3}{l}{Baseline (GHOST)}                                          & 83.4                           \\
\cmidrule{2-5}
                          & Feedforward      & {\IT}          & 5                                    & \gray{94}                     \\
                          & Naive Bayes      & {\IT}          & 5                                    & \gray{94}                              \\
                          & Logistic         & {\IT}          & 5                                    & \gray{94}                     \\
\cmidrule{2-5}
                          & Best (NB)        & Random              & 5                                    & \gray{94}                 \\
                          & Best (NB)        & BOHB                & 5                                    & \gray{94}                 \\
                          & Best (NB)        & BOHB                & 30                                   & \gray{94}                 \\
                          & Best             & DEHB                & 30                                   & 31.2 \\
\midrule
\multirow{8}{*}{log4j}    & \multicolumn{3}{l}{Baseline (GHOST)}                                          & 65.9                           \\
\cmidrule{2-5}
                          & Feedforward      & {\IT}          & 5                                    & \gray{79.6}                            \\
                          & Naive Bayes      & {\IT}          & 5                                    & \gray{79.6}                   \\
                          & Logistic         & {\IT}          & 5                                    & \gray{79.6}                   \\
\cmidrule{2-5}
                          & Best (NB)        & Random              & 5                                    & \gray{79.6}               \\
                          & Best (NB)        & BOHB                & 5                                    & \gray{79.6}               \\
                          & Best (NB)        & BOHB                & 30                                   & \gray{79.6}               \\
                          & Best             & DEHB                & 30                                   & 12.7 \\
\midrule
\multirow{8}{*}{synapse}  & \multicolumn{3}{l}{Baseline (GHOST)}                                          & \gray{84.6}                  \\
\cmidrule{2-5}
                          & Feedforward      & {\IT}          & 5                                    & \gray{84.4}                            \\
                          & Naive Bayes      & {\IT}          & 5                                    & \gray{84.4}                   \\
                          & Logistic         & {\IT}          & 5                                    & \gray{84.4}                   \\
\cmidrule{2-5}
                          & Best (NB)        & Random              & 5                                    & \gray{84.4}               \\
                          & Best (NB)        & BOHB                & 5                                    & 74.5                        \\
                          & Best (NB)        & BOHB                & 30                                   & 0                           \\
                          & Best             & DEHB                & 30                                   & 56.2 \\
\midrule
\multirow{8}{*}{velocity} & \multicolumn{3}{l}{Baseline (GHOST)}                                          & 50.3                            \\
\cmidrule{2-5}
                          & Feedforward      & {\IT}          & 5                                    & 55.4                            \\
                          & Naive Bayes      & {\IT}          & 5                                    & \gray{79.8}                   \\
                          & Logistic         & {\IT}          & 5                                    & \gray{79.8}                   \\
\cmidrule{2-5}
                          & Best (NB)        & Random              & 5                                    & \gray{79.8}               \\
                          & Best (NB)        & BOHB                & 5                                    & \gray{79.8}               \\
                          & Best (NB)        & BOHB                & 30                                   & \gray{79.8}               \\
                          & Best             & DEHB                & 30                                   & 54.7 \\
\midrule
\multirow{8}{*}{xalan}    & \multicolumn{3}{l}{Baseline (GHOST)}                                          & 66                              \\
\cmidrule{2-5}
                          & Feedforward      & {\IT}          & 5                                    & \gray{68.3}                   \\
                          & Naive Bayes      & {\IT}          & 5                                    & \gray{68.3}                   \\
                          & Logistic         & {\IT}          & 5                                    & \gray{68.3}                   \\
\cmidrule{2-5}
                          & Best (Logistic)  & Random              & 5                                    & \gray{68.2}               \\
                          & Best (Logistic)  & BOHB                & 5                                    & 0                           \\
                          & Best (Logistic)  & BOHB                & 30                                   & \gray{68.3}               \\
                          & Best             & DEHB                & 30                                   & 45.5 \\
\bottomrule
\end{tabular}
\caption{{\bf RQ1 results:}  Defect prediction results.
For Algorithm~\ref{alg:smoothie}, we use $N_1,N_2$ =30, 5.
Anything that ties with the statistically best results are shown in \fcolorbox{\graycolor}{\graycolor}{\textbf{bold}}
(and such ties
are determined using the statistics of \S\ref{sec:stats}).
}
\label{tab:defect}
\end{table}

%% file: bbo.tex
\begin{table} 
\centering
\caption{{\bf RQ1 results:}  Results on the NeurIPS 2020 Black-Box Optimization Challenge public datasets.
For Algorithm~\ref{alg:smoothie}, we use $N_1,N_2 =30, 5$.
Anything that ties with the statistically best results are shown in \fcolorbox{\graycolor}{\graycolor}{\textbf{bold}} (and such ties
are determined using the statistics of \S\ref{sec:stats}).
}
\label{tab:bbo}
\scriptsize
\begin{tabular}{llr}
\toprule
\textbf{Dataset}                 & \textbf{HPO method} & \multicolumn{1}{l}{\textbf{Score}} \\
\midrule
\multirow{5}{*}{breast} & Baseline (HEBO) & 90.84 \\
                        & {\IT}          & \gray{94.32}                     \\
                        & Random search       & 90.45                              \\
                        & BOHB                & \gray{92.37} \\
                        & DEHB                & 91.72 \\
\midrule
\multirow{5}{*}{digits} & Baseline (HEBO) & 95.24 \\
                        & {\IT}          & 89.51                                  \\
                        & Random search       & 87.39                              \\
                        & BOHB                & 91.50 \\
                        & DEHB                & \gray{96.87} \\
\midrule
\multirow{5}{*}{iris}   & Baseline (HEBO) & 78.27 \\
                        & {\IT}          & 87.64                    \\
                        & Random search       & 82.98                              \\
                        & BOHB                & \gray{92.30} \\
                        & DEHB                & 80.78 \\
\midrule
\multirow{5}{*}{wine}   & Baseline (HEBO) & 74.93 \\
                        & {\IT}          & \gray{93.34}                     \\
                        & Random search       & 76.59                              \\
                        & BOHB                & 81.68
                                \\
                        & DEHB                & 81.73 \\
\midrule
\multirow{5}{*}{diabetes} & Baseline (HEBO) & 98.67 \\
                        & {\IT}        & 98.65                                                  \\
                        & Random search       & 97.16                              \\                                      
                        & BOHB                & 97.86 \\
                        & DEHB                & \gray{99.40}
                                \\
\bottomrule
\end{tabular} 
\end{table}

%% file: biography.tex
\begin{IEEEbiography}[{\includegraphics[width=.9in,clip,keepaspectratio]{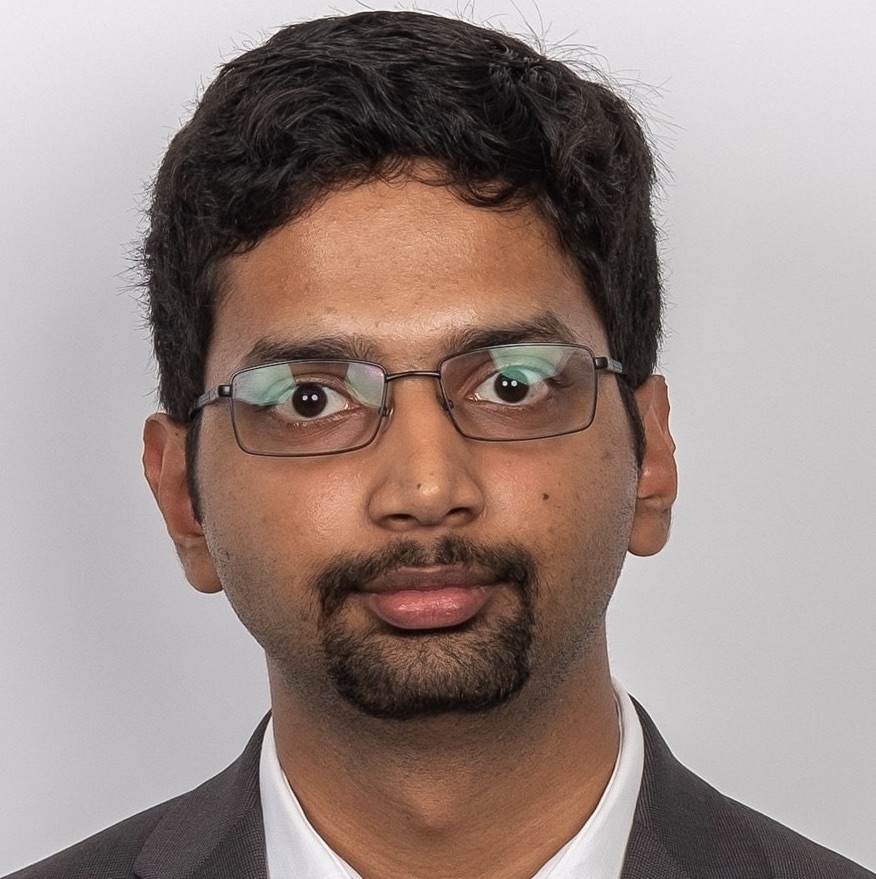}}] {Rahul Yedida} earned his PhD in Computer Science at NC State University and works as a Senior Data Scientist at LexisNexis. His research interests include automated software engineering and machine learning. For more information, please see \url{https://ryedida.me}.
\end{IEEEbiography}
\vspace{-20mm}
\begin{IEEEbiography}[{\includegraphics[width=.9in,clip,keepaspectratio]{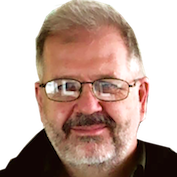}}]{Tim Menzies}  (ACM Fellow, IEEE Fellow, ASE Fellow, Ph.D., UNSW, 1995) is a full Professor in Computer Science at North Carolina State University. He is the director of the   Irrational Research Lab (mad scientists r'us) and the author of over 300 publications (refereed) with 24,000 citations and an h-index of 74. He has graduated 22 Ph.D. students, and has been a lead researcher on projects for NSF, NIH, DoD, NASA, USDA  and private companies (total funding of \$19+ million). Prof. Menzies is the editor-in-chief of the Automated Software Engineering journal and associate editor of TSE and other leading SE journals. For more, see  \url{https://timm.fyi}.
\end{IEEEbiography}

%% file: proofs.tex
\section*{Proofs}
\label{sec:proofs}

\begin{theorem}[Smoothness for Gaussian Naive Bayes]
    For Naive Bayes, the {\smoothness} is given by
    \[
        \beta = \sup \lVert \boldsymbol\Sigma^{-1} \mathcal{A}G - \frac{1}{2} \boldsymbol\Sigma^{-1} \mathcal{A} \boldsymbol\Sigma^{-1} + G \mathcal{A} \boldsymbol\Sigma^{-1} \rVert
    \]
    where $\boldsymbol\Sigma$ is the covariance matrix, $\mathcal{A}$ is the fourth-order identity tensor, and $G = \dfrac{\partial\ell}{\partial\boldsymbol\Sigma}$, where $\ell$ is the log-likelihood.
\end{theorem}

\begin{proof}
We use the results of \citet{4647791} here. For Gaussian Discriminant Analysis (which Naive Bayes is a special case of), the log-likelihood is given by:

\begin{equation}\label{gda}
\begin{array}{rcl}
\ell(\phi, \boldsymbol \mu_0, \boldsymbol \mu_1, \boldsymbol \Sigma) &=\sum_{i=1}^m \bigl\{ y^{(i)} \log \phi + (1 - y^{(i)}) \log (1 - \phi)   \\~\\  & -\frac{1}{2} \log |\boldsymbol \Sigma|+ C + (\textbf{x}^{(i)}-\boldsymbol \mu)^T \boldsymbol \Sigma^{-1} (\textbf{x}^{(i)} - \boldsymbol \mu) \bigr\}
\end{array}
\end{equation}

where $\phi$ is the parameter of the Bernoulli distribution; $\mu_0$ and $\mu_1$ are the mean vectors for the two classes and $\Sigma$ is the covariance matrix of the independent attributes. Because there are several parameters, we show that the {\smoothness} is \textit{constant} with respect to all except the covariance matrix, and therefore we use the {\smoothness} with respect to the covariance matrix in our experiments.

Taking a gradient of \eqref{gda} with respect to  $\phi$ gives us:

\[
\begin{aligned}
    \nabla_\phi \ell &= \frac{\sum y^{(i)}}{\phi} - \frac{\sum (1-y^{(i)})}{1-\phi} \\
    \nabla^2_\phi \ell &= -\frac{\sum y^{(i)}}{\phi^2} - \frac{\sum (1-y^{(i)})}{(1-\phi)^2}
\end{aligned}
\]

The second derivative has an extremum when  $\nabla_\phi^3 \ell = 0$:

\[
\begin{aligned}
    \nabla^3_\phi \ell &= \frac{2\sum y^{(i)}}{\phi^3} - \frac{2\sum (1-y^{(i)})}{(1-\phi)^3} = 0 \\
    \frac{\sum y^{(i)}}{\phi^3} &= \frac{\sum (1-y^{(i)})}{(1-\phi)^3} \\
    m\phi^3 &= \left( \sum y^{(i)} \right) (1-3\phi + 3\phi^2)
\end{aligned}
\]

But $\sum y^{(i)} = km$ for some $0 \leq k \leq 1$. Hence,
 
$$
\phi^3 = k(1-3\phi+3\phi^2)
$$

which has constant solutions. For example, if $k = \frac{1}{2}$ then $\phi=\frac{1}{2}$ and $\nabla_\phi^2 E \leq -4m$. Similarly, if $k=1$ then $\phi=1$ and $\nabla_\phi^2 E \leq -m$

From \eqref{gda},   computing   gradient with respect to
$\mu$  yields

{\small \[
    \begin{aligned} 
        \nabla_{\mu_0} \ell &= \nabla_{\mu_0} \sum_{i=1}^m (\textbf{x}^{(i)}-\boldsymbol \mu)^T \boldsymbol \Sigma^{-1} (\textbf{x}^{(i)} - \boldsymbol \mu) \\ 
        &= \nabla_{\mu_0} \sum_{i=1}^m \text{tr } (\textbf{x}^{(i)}-\boldsymbol \mu)^T \boldsymbol \Sigma^{-1} (\textbf{x}^{(i)} - \boldsymbol \mu)I \\ 
        &= -\sum_{i=1}^m \left( I(\textbf{x}^{(i)}-\boldsymbol \mu_0) \boldsymbol \Sigma^{-1} + I(\textbf{x}^{(i)}-\boldsymbol \mu_0) (\boldsymbol \Sigma^{-1})^T \right) [y^{(i)} = 0] \\ 
        &= -\sum_{i=1}^m \left( (\textbf{x}^{(i)}-\boldsymbol \mu_0) \boldsymbol \Sigma^{-1} + (\textbf{x}^{(i)}-\boldsymbol \mu_0) \boldsymbol \Sigma^{-1} \right) [y^{(i)} = 0] \\ &= -2 \boldsymbol \Sigma^{-1} \sum_{i=1}^m (\textbf{x}^{(i)} - \boldsymbol \mu_0) [y^{(i)} = 0] 
    \end{aligned}
\]}

where we used the identity $\nabla_A \text{tr } ABA^TC = CAB + C^TAB^T$ along with the trace trick. Clearly, the second derivative is constant with respect to $\boldsymbol \mu_0$.

Finally, we compute the smoothness with respect to $\boldsymbol\Sigma$. For the following, let $\langle x, y \rangle$ denote the Frobenius inner product defined as 
\[
\langle A, B \rangle = \sum\limits_{i=1}^m \sum\limits_{j=1}^n A_{ij} B_{ij} = \text{tr } A^T B
\]

Let $\boldsymbol w = ( \boldsymbol x - \boldsymbol \mu_i)$ (where $\mu_i$ is the mean
for class $i$). Then,
\[
\begin{aligned}
    d\ell &= \langle \boldsymbol w \boldsymbol w^T, (d\boldsymbol\Sigma)^{-1} \rangle - \frac{1}{2} \langle \boldsymbol\Sigma^{-1}, d\boldsymbol\Sigma \rangle \\
    &= \langle \boldsymbol w \boldsymbol w^T, -\boldsymbol\Sigma^{-1} d\boldsymbol\Sigma \boldsymbol\Sigma^{-1} \rangle - \frac{1}{2} \langle \boldsymbol\Sigma^{-1}, d\boldsymbol\Sigma \rangle \\
    &= -\left\langle \boldsymbol\Sigma^{-1} \boldsymbol w \boldsymbol w^T \boldsymbol\Sigma^{-1} + \frac{1}{2} \Sigma^{-1}, d\Sigma \right\rangle \\
    \dfrac{\partial\ell}{\partial\boldsymbol\Sigma} &= -\boldsymbol\Sigma^{-1} \left(\boldsymbol w \boldsymbol w^T + \frac{1}{2} \boldsymbol\Sigma \right)\boldsymbol\Sigma^{-1} = G
\end{aligned}
\]

Next, we compute the differential of $G$, which we use to compute the Hessian.

\[
\begin{aligned}
    dG &= -d\boldsymbol\Sigma^{-1} \left(\boldsymbol w \boldsymbol w^T + \frac{1}{2} \boldsymbol\Sigma \right)\boldsymbol\Sigma^{-1} - \boldsymbol\Sigma^{-1} \left( \frac{1}{2} d\boldsymbol\Sigma \right) \boldsymbol\Sigma^{-1} \\ &- \boldsymbol\Sigma^{-1} \left(\boldsymbol w \boldsymbol w^T + \frac{1}{2} \boldsymbol\Sigma \right)d\boldsymbol\Sigma^{-1} \\
    &= \boldsymbol\Sigma^{-1} d\boldsymbol\Sigma G -\frac{1}{2} \boldsymbol\Sigma^{-1} d\boldsymbol\Sigma \boldsymbol\Sigma^{-1} + G d\boldsymbol\Sigma \boldsymbol\Sigma^{-1} \\ 
    &= \langle \boldsymbol\Sigma^{-1} \mathcal{A}G - \frac{1}{2} \boldsymbol\Sigma^{-1} \mathcal{A} \boldsymbol\Sigma^{-1} + G \mathcal{A} \boldsymbol\Sigma^{-1}, d\boldsymbol\Sigma \rangle \\
    \dfrac{\partial G}{\partial \boldsymbol\Sigma} &= \boldsymbol\Sigma^{-1} \mathcal{A}G - \frac{1}{2} \boldsymbol\Sigma^{-1} \mathcal{A} \boldsymbol\Sigma^{-1} + G \mathcal{A} \boldsymbol\Sigma^{-1} = \mathcal{H}
\end{aligned}
\]

where $\mathcal{H}$ is the Hessian and $\mathcal{A}$ is the fourth-order identity tensor so that $\mathcal{A}_{ijkl} = \delta_{ik}\delta_{jl}$. In the above, the fact that $\Sigma$ and $G$ are symmetric was used in various steps.

Although the \textit{loss} to minimize would be defined as the negative log-likelihood, the negative sign would be removed by the use of a norm.
\end{proof}

\begin{theorem}[$L_2$ regularization]
    For a learner using an $L_2$ regularization term $\frac{\lambda}{2} \lVert \boldsymbol w \rVert^2_2$, the increase in {\smoothness} is $\lambda$. More generally, if a Tikhonov regularization term $\lVert \boldsymbol \Gamma \boldsymbol w \rVert^2_2$ is used, the increase in {\smoothness} is $2\lVert \boldsymbol \Gamma \rVert^2_F$.
\end{theorem}

\begin{proof}
    For $L_2$ regularization, the regularization term is $E = \frac{\lambda}{2}\lVert \boldsymbol w \rVert_2^2$. We therefore have $\nabla_{\boldsymbol w} E = \lambda \lVert \boldsymbol w \rVert_2$ and so $\lVert \nabla_{\boldsymbol w_1} E(\boldsymbol w_1) - \nabla_{\boldsymbol w_2} E(\boldsymbol w_2) \rVert \leq \lambda \lVert \boldsymbol w_1 - \boldsymbol w_2 \rVert$.
    For the general Tikhonov case,
    \[
        \begin{aligned}
            E(\boldsymbol w) &= \lVert \boldsymbol \Gamma \boldsymbol w \rVert^2_2 \\
            &= \boldsymbol w^T \boldsymbol \Gamma^T \boldsymbol \Gamma \boldsymbol w \\
            dE &= d\boldsymbol w^T \boldsymbol \Gamma^T \boldsymbol \Gamma \boldsymbol w + \boldsymbol w^T d\boldsymbol \Gamma^T \boldsymbol \Gamma \boldsymbol w + \boldsymbol w^T \boldsymbol \Gamma^T d\boldsymbol \Gamma + \boldsymbol w^T \boldsymbol \Gamma^T \boldsymbol \Gamma d\boldsymbol w \\
            &= (d\boldsymbol w)^T \boldsymbol \Gamma^T \boldsymbol \Gamma \boldsymbol w + \boldsymbol w^T \boldsymbol \Gamma^T \boldsymbol \Gamma (d\boldsymbol w) \\
            &= \boldsymbol w^T \boldsymbol \Gamma^T \boldsymbol \Gamma (d\boldsymbol w) + \boldsymbol w^T \boldsymbol \Gamma^T \boldsymbol \Gamma (d\boldsymbol w) \\
            \nabla_{\boldsymbol w} E &= 2\boldsymbol w^T \boldsymbol \Gamma^T \boldsymbol \Gamma = 2 \boldsymbol w^T \lVert \boldsymbol \Gamma \rVert^2_F
        \end{aligned}
    \]

    so that
    \[
        \frac{\lVert \nabla E(\boldsymbol w_1) - \nabla E(\boldsymbol w_2) \rVert}{\lVert \boldsymbol w_1 - \boldsymbol w_2 \rVert} = 2\lVert \boldsymbol \Gamma \rVert^2_F
    \]
\end{proof}

\subsection*{Why 30 Initial Samples Suffices}

Understanding why 30 samples are enough is a problem similar to the coupon collector's problem. Consider the following problem, which is equivalent to covering the entire hyper-parameter space. 
    
    \begin{formal}
    Suppose we have a space $[a, b] \subset \mathbb{R}$, and one iteratively chooses a point $x_i \overset{i.i.d}{\sim} \mathcal{U}[a, b]$. In doing so, one ``covers'' the range $(x_i-k, x_i+k)$ for some $k \ll b - a$. Define the ``coverage'' as the portion of the range that has been covered. Suppose $p$ points have been chosen. What are the bounds on the expected coverage?
    \end{formal}

    The probability of a point \textit{not} being covered by any of the previously chosen $p$ points is, by a union bound, lower bounded by $1 - \frac{2pk}{b-a} \approx \exp\left( -\frac{2pk}{b-a} \right)$. Therefore, the probability of it being covered by any of the $p$ points is upper bounded by $1 - \exp\left( -\frac{2pk}{b-a} \right)$.

    In the multi-dimensional hyper-parameter optimization setting, this is a trivial extension involving hyperrectangles. Specifically, the fraction of the range covered is upper bounded by:
    \begin{equation}
    1 - \exp\left( -\frac{p (2k)^d}{(b-a)^d} \right)
    \label{eq:pcover:simple}
    \end{equation}
    where $d$ is the number of dimensions. More generally, if we instead assume that each dimension has size $L_i$ and choosing a point covers $2k_i$ along that dimension, this expression becomes
    \begin{equation}
    1 - \exp\left( -p \prod\limits_{i=1}^d \frac{2k_i}{L_i} \right)
    \label{eq:pcover:complex}
    \end{equation}

    To get a lower bound, let $X_p$ denote the coverage at $p$ samples. We are interested in $\mathbb{E}[X_p]$. We can then form a recurrence relation based on the idea that we need a distance of at least $k$ from any other point to get the maximum coverage of $2k$ from the new point:
    \begin{align*}
        \mathbb{E}[X_p - X_{p-1} | X_{p-1}] &\geq 2k ((b-a) - X_{p-1} - k) \\
        \mathbb{E}[X_p - X_{p-1}] &\overset{(i)}{=} \mathbb{E}_{X_{p-1}}[\mathbb{E}[X_p - X_{p-1} | X_{p-1}]] \\
        &\geq \mathbb{E}[2k ((b-a) - X_{p-1} - k)] \\
        \mathbb{E}[X_p] &\geq (1-2k)\mathbb{E}[X_{p-1}] + 2k(b-a-k)
    \end{align*}
    where (i) is the tower rule. We can now solve this recurrence relation, using $E_p$ to denote $\mathbb{E}[X_p]$. Trivially, $E_1 = 2k$.
    \begin{align*}
        E_p &\geq (1-2k)E_{p-1} + 2k(b-a-k) \\
        &= (1-2k)^{t}E_{p-t} + 2k(b-a-k) \sum\limits_{i=0}^{t-1} (1-2k)^i \\
        &= (1-2k)^t E_{p-t} + 2k(b-a-k) \frac{1 - (1-2k)^t}{2k} \\
        &= (1-2k)^t E_{p-t} + (b-a-k)\left( 1 - (1-2k)^t \right) \\
        &= (1-2k)^{p-1}(2k) + (b-a-k)\left(1 - (1-2k)^{p-1} \right) \\
        &= (b-a-k) - (1-2k)^{p-1}(b-a-3k)
    \end{align*} \qed

    Both these bounds are exponential, so that exploring more options in a randomized fashion has very quickly diminishing returns. Additionally, we add that not all hyper-parameters have equal value: for example, some hyper-parameters are significantly more important than others \cite{van2018hyperparameter, bergstra2011algorithms}.